\documentclass[twocolumn,aps,pra,superscriptaddress,groupedaddress]{revtex4-1}
\usepackage[latin9]{inputenc}
\usepackage{amsmath}
\usepackage{amssymb}
\usepackage{mathtools}
\usepackage{graphicx}
\usepackage{braket}
\usepackage{hyperref}
\usepackage{paralist}
\usepackage{color}
\usepackage{mathrsfs}

\makeatletter
\@ifundefined{textcolor}{}
{%
 \definecolor{BLACK}{gray}{0}
 \definecolor{WHITE}{gray}{1}
 \definecolor{RED}{rgb}{1,0,0}
 \definecolor{GREEN}{rgb}{0,1,0}
 \definecolor{BLUE}{rgb}{0,0,1}
 \definecolor{CYAN}{cmyk}{1,0,0,0}
 \definecolor{MAGENTA}{cmyk}{0,1,0,0}
 \definecolor{YELLOW}{cmyk}{0,0,1,0}
}

\usepackage{bm}\usepackage{amsthm}\usepackage{epstopdf}\usepackage{txfonts}

\newtheorem*{lemma*}{Lemma}
\newtheorem{prop}{Proposition}
\newtheorem{definition}{Definition}
\theoremstyle{definition}
\newtheorem{example}{Example}

\DeclareMathOperator*{\Tr}{Tr}
\DeclareMathOperator*{\Rank}{Rank}
\DeclareMathOperator*{\Ran}{Ran}
\DeclareMathOperator*{\Span}{Span}
\DeclareMathOperator*{\Ker}{Ker}

\DeclareMathOperator*{\Real}{Re}

\DeclareMathOperator{\diag}{diag}
\DeclareMathOperator{\Prob}{Prob}
\DeclareRobustCommand{\openzero}{\leavevmode\hbox{0\kern-.55em0}}

\newcommand{\bs}{\boldsymbol}
\newcommand{\e}{\mathrm{e}}
\newcommand{\im}{\mathrm{i}}
\renewcommand\bra[1]{{\langle{#1}|}}
\renewcommand\ket[1]{{|{#1}\rangle}}

\makeatother

\begin{document}
\title{Coherence generating power of quantum dephasing processes}

\date{\today}

\author{Georgios Styliaris}
\email [e-mail address: ]{styliari@usc.edu}
\author{Lorenzo Campos Venuti}
\author{Paolo Zanardi}

\affiliation{Department of Physics and Astronomy, and Center for Quantum Information
Science \& Technology, University of Southern California, Los Angeles,
CA 90089-0484}

\begin{abstract}
We provide a quantification of the capability of various quantum dephasing processes to generate coherence out of incoherent states. The measures defined, admitting computable expressions for any finite Hilbert space dimension, are based on probabilistic averages and arise naturally from the viewpoint of coherence as a resource. We investigate how the capability of a dephasing process (e.g., a non-selective orthogonal measurement) to generate coherence depends on the relevant bases of the Hilbert space over which coherence is quantified and the dephasing process occurs, respectively.
We extend our analysis to include those Lindblad time evolutions which, in the infinite time limit, dephase the system under consideration and calculate their coherence generating power as a function of time. We further identify specific families of such time evolutions that, although dephasing, have optimal (over all quantum processes) coherence generating power for some intermediate time. Finally, we investigate the coherence generating capability of random dephasing channels.
\end{abstract}

\maketitle

\section{Introduction} \label{intro_section}

One of the main distinctive features of quantum theory is the \textit{superposition principle}. According to it, physical states of a quantum system can be expressed as linear combinations of other quantum states and different bases, usually associated with eigenstates of observables, yield different expansions. The presence of accessible relative phases between the different branches is known as \textit{quantum coherence} and gives rise to quantum interference phenomena, lying in the heart of theory \cite{bohm1951quantum}. Quantum coherence, besides being a integral part of the quantum theory, constitutes also an important ingredient, for example, in quantum metrology \cite{giovannetti2001quantum,giovannetti2004quantum,demkowicz2014using}, quantum computation \cite{nielsen2002quantum} and quantum error correction \cite{lidar2013quantum}, quantum thermodynamics \cite{brandao2013resource,gour2015resource,PhysRevX.5.021001} and quantum biological processes \cite{engel2007evidence,mohseni2008environment,huelga2013vibrations}.

On the other hand, important classes of dynamics in open quantum systems as well as various measurement processes lead to \textit{dephasing} of the system under consideration (see, e.g., \cite{breuer2002theory} and \cite{wiseman2009quantum}). Dephasing processes are linked to loss of information associated with the relative phases between the branches of the wavefunction. Nevertheless, dephasing of a quantum state does not necessarily imply total loss of its quantum coherence, since both dephasing and coherence are notions well-defined only with respect to specific bases in the Hilbert space which can, in general, be different.


\textit{The main purpose of this work is to quantify the capability of various dephasing processes to generate quantum coherence out of incoherent states}. We investigate how the efficiency of a dephasing process (e.g., a non-selective orthogonal measurement) to generate coherence depends on the associated bases in the Hilbert space and study its maximization. The situation of a dephasing process occurring over a random basis is also examined. We further consider quantum evolutions described by the Lindblad master equation which lead to dephasing of the system under consideration and examine how their capability to generate coherence varies as a function of time. Remarkably, we find that there exist time instances over which certain such dephasing evolutions can generate coherence as well as the optimal unitary processes.

In this work we are interested in the \textit{capability} of (particular classes of) quantum operations to produce coherence, rather that just in the amount of coherence being present in quantum states. In general, no unique formulation in resource theories exists for this quantification and different approaches, encapsuling different aspects, exist (see, e.g., \cite{zanardi2000entangling,mani2015cohering,garcia2015note,bu2017cohering,dana2017resource,korzekwa2017coherifying}). Here, we adopt the relevant definition of \textit{coherence generating power} of quantum operations based on probabilistic averages \cite{coherence_1,coherence_2,zanardi2017quantum} (summarized in \autoref{CGP_section}).

This article is organized as follows. In \autoref{definitions_section} we give the preliminary mathematical definitions and set the notation. In \autoref{resource_theories_section} we recall basic aspects of the resource theory of quantum coherence while in \autoref{CGP_section} we present the main definitions of the \textit{coherence generating power} formalism (following \cite{coherence_1,coherence_2}), which is a way of extending the quantification of coherence from states to quantum operations, and expand on them. The main body of the article is \autoref{main_section}, where we first distinguish between different classes of dephasing processes (\autoref{maximal_partial_definitions_section}) and then we quantify the capability of such processes to generate coherence (\autoref{maximal_section} and \autoref{partial_section}). In \autoref{Lindblad_section} we investigate dephasing time evolutions obeying the Lindblad master equation and study the coherence generating power of those evolutions as a function of time (\autoref{CGP_Lindblad_section}) as well as specific families of those evolutions attaining optimal coherence generating power (\autoref{Lindblad_maximal_section}). Finally, in \autoref{random_section} we study the coherence generating capability of random dephasing processes. In  \autoref{conclusions_section} we conclude.

\section{Setting the stage}

\subsection{Basic definitions} \label{definitions_section}

Let $ \{\ket{i}\}_{i=1}^d$ be an orthonormal basis of the Hilbert space ${\mathcal H}\cong{\mathbb{C}}^d$ and $\mathbf \{ P_i \coloneqq \ket{i}\bra{i} \}_{i=1}^d$ be the associated family of rank-1 orthogonal projectors. We consider the operator space $ \mathcal B (\mathcal H)$ over $\mathcal H$ as a Hilbert space equipped with the Hilbert-Schmidt scalar product $\braket{X,Y} \coloneqq \Tr (X^\dagger Y)$ and norm $\|X\|_2:=\sqrt{\braket{X,X}}$. The (Shatten) 1-norm of operator $X$ is defined as $\left\| X \right\|_1 \coloneqq \Tr \left( \sqrt{X^\dagger X} \right) = \sum_{i=1}^d s_i$ (where $\{s_i\}_{i=1}^d$ are the singular values of $X$) while $\| X \|_\infty $ denotes the operator (spectral) norm, i.e.,  $\| X \|_\infty \coloneqq \max_i \left( s_i \right)$.

The above construction can be extended to the superoperator space $ \mathcal B \left( \mathcal B\left( \mathcal H \right)  \right)$, which can be similarly equipped with a (Hilbert-Schmidt over the Hilbert space $\mathcal B\left( \mathcal H \right)$) scalar product $\braket{\mathcal X,\mathcal Y} \coloneqq \Tr (\mathcal X^\dagger \mathcal Y)$ and a 2-norm $\|\mathcal X\|_2 \coloneqq \sqrt{\braket{\mathcal X,\mathcal X}}$ (where $\mathcal X , \mathcal Y \in \mathcal B \left( \mathcal B \left( \mathcal H \right)  \right)\,$) \footnote{Notice that we use the same symbol $|| ( \cdot )|| _{2}$ both for the operator 2-norm and the superoperator 2-norm. Which definition is applicable should be clear from the context.}.
The 1-1 induced norm is denoted as $\left\| \mathcal X \right\|_{1,1} \coloneqq \sup_{\left\| A \right\|_1 =1} \left( \left\| \mathcal X A \right\|_1 \right)$. The 1-1 norm is unstable under tensorization \cite{kitaev2002classical}, i.e., in general $\left\| \mathcal X \otimes \mathcal Y \right\|_{1,1} \ne \left\| \mathcal X \right\|_{1,1} \left\| \mathcal Y \right\|_{1,1}$. Nevertheless the diamond norm, which can be defined as $\left\| \mathcal X \right\|_\diamondsuit \coloneqq  \left\| \mathcal X \otimes I_d  \right\|_{1,1}$ \cite{kitaev1997quantum}, satisfies $\left\| \mathcal X \otimes \mathcal Y \right\|_{\diamondsuit} = \left\| \mathcal X  \right\|_{\diamondsuit}  \left\| \mathcal Y \right\|_{\diamondsuit} $ ($I_d$ above denotes the identity superoperator over $\mathcal H \cong {\mathbb{C}}^d$).
We define as physically valid quantum operations $\mathcal E$ over the set of density operators $\mathcal S (\mathcal H) \subset \mathcal B(\mathcal H)$ all the linear, Completely Positive (CP) and Trace Preserving (TP) maps $\mathcal E: \mathcal S(\mathcal H) \to \mathcal S(\mathcal H)$.

Given a complete set of orthogonal (not necessarily rank-1) projectors $B=\{ \Pi_i\}_i$ (i.e., $\Pi_i = \Pi_i^\dagger$, $\Pi_i \Pi_j =  \Pi_i \, \delta_{ij}$, $\sum_i \Pi_i = I$) we define the \textit{$B$-dephasing superoperator} as
\begin{gather}
\mathcal D_B (\cdot) \coloneqq \sum_i \Pi_i (\cdot) \Pi_i \,\;,
\end{gather}
which is an orthogonal projector over $\mathcal B(\mathcal H)$. The complementary projector is denoted $\mathcal Q_B \coloneqq I - \mathcal D_B$. Every orthonormal basis $ \{\ket{i}\}_{i=1}^d$ of $\mathcal H$ has an associated dephasing superoperator $\mathcal D_B$, where $B = \{ P_i \}_{i=1}^d$. The range of a linear operator $X$ is denoted as $\Ran (X)$.

For a $d$-dimensional probability vector $\bs p$ (i.e., $p_i \ge 0 $, $\sum_{i=1}^d p_i = 1$) we denote its \textit{Shannon entropy} as $H (\bs p) \coloneqq - \sum_{i=1}^d  p_i \ln(p_i)$. We use $S(\rho) \coloneqq - \Tr \left( \rho \ln \rho \right)$ for the \textit{von Neumann entropy} of $\rho \in \mathcal S(\mathcal H)$.

\subsection{Resource theory of quantum coherence}  \label{resource_theories_section}


The main idea behind quantum resource theories is simple: a subset of the physical states of the quantum system under consideration is distinguished (called \textit{free states}), as well as a subset of its possible quantum evolutions (called \textit{free operations}). The free states contain no resource (by definition) and, for the resource theory to be consistent, free operations should map free states to free states.
Any proper resource quantifier is a function mapping quantum states to non-negative real numbers with the properties that
\begin{inparaenum}[(i)]
\item the amount of resource of a free state vanishes, and
\item the amount of resource contained in any state cannot increase under the action of free operations.
\end{inparaenum}
Such functions are hence also called \textit{resource monotones}.

In the resource theory of coherence the set of free states $I_B$, which are called \textit{incoherent states}, is defined (with respect to a basis) as the image (over quantum states) of the associated $B$-dephasing superoperator: $I_B = \Ran \left( \mathcal D_B \right) $ where $B = \{ P_i\}_{i=1}^d$, i.e., a state $\rho$ is incoherent if and only if
$\rho = \sum_i p_i P_i$ with $\{p_i\}_{i=1}^d$ any probability distribution.

The set of free operations $\mathcal I_B$  has to be compatible with the set of free states, by ensuring no resource can be generated by the action of free operation on free states, i.e., if $\mathcal W \in \mathcal I_B$ is free then $\mathcal W (\rho) \in I_B$ for any $\rho \in I_B$. This is the minimal requirement of the theory for consistency and gives raise to the largest class of free operations, known as Maximally Incoherent Operations (MIO) \cite{aberg2006quantifying} which are quantum operations $\mathcal W$ such that $\mathcal W(I_B) \subseteq I_B$. Several alternative subclasses of free operations have been defined and investigated (see, e.g., \cite{streltsov2016quantum}), here we mention just a few:
\begin{itemize}
  \item  The subclass of \textit{Incoherent Operators} (IO) \cite{baumgratz2014quantifying} consists of the CPTP maps admitting a set of Kraus operators \cite{preskill1998lecture} $\{K_n\}_n$ such that, for all $\rho \in I_B$, $K_n \rho K_n^\dagger / \Tr(K_n \rho K_n^\dagger) \in I_B$.
  \item The subclass of \textit{Dephasing-covariant Incoherent Operators} (DIO) \cite{chitambar2016comparison,Marvian2016Speakable} contains the operators $\mathcal W$ such that $\left[ \mathcal W , \mathcal D_B \right] = 0 $.
  \item The subclass of Strictly Incoherent Operations (SIO) \cite{winter2016operational} contains the operators in IO that in addition fulfill $K_n^\dagger \rho K_n / \Tr(K_n^\dagger \rho K_n) \in I_B$ for any $\rho \in I_B$.
  \item Finally, Genuinely Incoherent Operators (GIO) \cite{Streltsov2016genuine} are the operations $\mathcal W$ that leave all incoherent states invariant, i.e., for any $\rho \in I_B$ it holds that $\mathcal W (\rho) = \rho$. From the definition it immediately follows that all GIO are in addition \textit{unital}.
\end{itemize}

A functional $c_B: \mathcal S(\mathcal H) \to \mathbb R _0^+$ is a coherence monotone if the following two properties hold: \begin{inparaenum}[(i)]
\item $c_B(\rho) = 0$ for all $\rho \in I _B$, and
\item $c_B(\mathcal W \rho) \le c(\rho)$ for all $\rho \in \mathcal S(\mathcal H)$ and $\mathcal W \in \mathcal I _B$.
\end{inparaenum}
Clearly, whether such a functional is a monotone or not depends on the set of free operations, e.g., a monotone of a subclass of MIO is not necessarily a monotone for MIO. The converse, however, is true: MIO monotones are monotones for all the possible subclasses. Monotones impose necessary conditions for interconversion of states under free operations, since $c_B(\rho_1) < c_B(\rho_2)$ suggests that it is impossible to convert $\rho_1 \rightarrow \rho_2$ under free operations. Monotones, therefore, \textit{quantify} how much resource a state contains -- this amount cannot increase under free operations.

\subsection{Coherence generating power of quantum operations} \label{CGP_section}

We begin by reviewing the definition of the \textit{Coherence Generating Power} (CGP) of a quantum operation,
as it was introduced in Ref.~\cite{coherence_1}. The idea pursued there, which allows transitioning from quantifying the amount of coherence in a quantum state $\rho$ to quantifying the ability of a quantum operation $\mathcal E$ to generate coherence, is simple: one imagines $\mathcal E$ acting on $\rho_{in} = \sum_i p_i P_i$, the latter being chosen at random from a uniform ensemble of ``input'' states, all of which are incoherent with respect to $B = \{ P_i \}_{i=1}^d$. Then one averages the amount of coherence contained in the processed state $\mathcal E (\rho_{in})$ over the ensemble of input states (i.e. $\{ p_i\}_{i=1}^d$ are treated as random variables), obtaining a quantifier characterizing $\mathcal E$. Clearly, this quantification of the ability of the quantum channel to generate coherence depends on the choice of the measure of (state) coherence $c_B$. The choice of the coherence monotone $c_B$ is far from unique and different choices are possible, depending on the set of free operations.

This approach is encapsuled in the following definition.
\begin{definition}
  The Coherence Generating Power (CGP) $C_B: \mathcal{E} \mapsto C_B(\mathcal E) \in  \mathbb R^+_0$ of a quantum operation $\mathcal E$ with respect to $B = \{ P_i \}_{i=1}^{d}$ and coherence measure $c_B$ is defined as
\begin{gather}
  C_B(\mathcal E) \coloneqq \int d \mu_{unif}(\bs p) \, c_B \left[ \mathcal E \left( \textstyle {\sum_i} p_i P_i \right) \right] , \label{CGP_formula}
\end{gather}
where $d \mu_{unif} (\bs p) \coloneqq \frac{1}{(d-1)!} \, \delta\left(\sum_i p_i -1\right) \prod_i d\!p_i$ is the uniform measure in the $(d-1)$-dimensional simplex.
\end{definition}

The $(d-1)$-dimensional simplex is the space of all possible d-tuples $\bs p = (p_1,\dotsc,p_d)$ with $p_i \ge 0$ and $\sum_i p_i = 1$, the points of which are in one to one correspondence with the diagonal elements of the incoherent input states $\rho_{in} = \sum_{i=1}^d p_i P_i$ (assuming a fixed $B$ with respect to which all input states are diagonal).

Before proceeding further by specifying the coherence measure $c_B$, we provide an alternative interpretation for the meaning of the CGP of a unitary quantum map $C_B\left( \mathcal U \right)$. Suppose we are interested in the following question: given a random pure state $\ket{\psi} \bra{\psi} \in \mathcal S(\mathcal H)$, what is the average coherence $c_B$ of the dephased quantum state $\mathcal D_{B'} \left(  \ket{\psi} \bra{\psi} \right)$, for some fixed bases $B$ and $B'$? For what follows, we assume the random pure states are distributed according to the Haar measure. Related matters were investigated in \cite{singh2016average,zhang2017average}.

As we will show now, the average coherence present after the dephasing of a random pure states is nothing else than the CGP of a corresponding unitary operator connecting the bases $B$ and $B'$.
\begin{prop}
  Let $B = \{ \ket{i} \bra{i}\}_{i=1}^d$ and $B' = \{ \ket{i'} \bra{i'} \}_{i=1}^d$ be complete families of rank-1 orthogonal projectors and $U \in U(d)$ be a unitary operator such that $\ket{i'} = U \ket{i}$ for all $i=1,\dotsc,d$. Then
  \begin{gather}
    \int d \mu_{Haar} (\psi) \, c_B \left( \mathcal D_{B'}   \ket{\psi} \bra{\psi}  \right) = C_B(\mathcal U) \,\;,
  \end{gather}
  where $\mathcal U (\cdot) = U (\cdot) U^\dagger$.
\end{prop}
\begin{proof}
  As it was shown in \cite{coherence_1}, one can equivalently treat the input states of the CGP definition Eq.~\eqref{CGP_formula} as Haar distributed pure states that are dephased in $B$, i.e.,
  \begin{gather}
  C_B(\mathcal E) = \int d \mu_{Haar}(\psi) \,  c_B \left[ \mathcal E \mathcal D_B \left( \ket{\psi} \bra{\psi} \right) \right] \,\;.
  \end{gather}
  The result therefore follows from $\mathcal D_{B'} = \mathcal U \mathcal D_B \mathcal U^\dagger$ and the fact that the Haar measure is unitarily invariant.
\end{proof}
Clearly, the above result holds for all valid coherence measures $c_B$. Therefore, besides suggesting an alternative interpretation for the CGP of unitary operators, it further allows applying any already known results for the CGP of unitary operators, e.g., from \cite{coherence_1,ZhangToAppear}.

Now we specify the coherence measure $c_B$. We examine two possible choices: the monotone arising from the Hilbert-Schmidt 2-norm and the relative entropy of coherence.


\subsubsection{Hilbert-Schmidt 2-norm based monotone}

The Hilbert-Schmidt operator 2-norm gives rise to the coherence measure
\begin{align}
  c_{2,B}(\rho) \coloneqq \min _{\sigma \in I_B} \| \rho - \sigma  \|_2^2 = \| \mathcal Q_B \rho  \|_2^2  = \sum_{i \ne j} | \rho_{ij} |^2 \,\;.
\end{align}
The functional $c_{2,B}$ is a coherence monotone in the (restrictive) class of GIO. More generally, it constitutes a monotone for any class of free operations if one restricts to \textit{unital} CPTP maps \cite{coherence_2} (such as dephasing processes considered later). Although specific to unital maps, the $c_{2,B}$ coherence quantifier allows for explicit computable formulas for the CGP in any finite Hilbert space dimension \cite{coherence_1}:
\begin{prop}
  Let $C_{2,B}(\mathcal E)$ be the coherence generating power (Eq.~\eqref{CGP_formula}) of the unital quantum channel $\mathcal E$ with coherence quantifier $c_B = c_{2,B}$. Then,
  \begin{enumerate}[(i)]
    \item \begin{align}
             C_{2,B}(\mathcal E) &= \frac{1}{d(d+1)}\sum_i \Bigl( \braket{\mathcal E  P_i , \mathcal E P_i} - \braket{\mathcal D_B \mathcal E  P_i , \mathcal D_B \mathcal E P_i}  \Bigr) \label{CGP_unitals}
          \end{align}
    \item If $\mathcal E$ is, in addition, normal ($[ \mathcal E , \mathcal E ^\dagger ] = 0$), then
          \begin{align}
             C_{2,B}(\mathcal E) = \frac{1}{2d(d+1)} \left\| \, \left[ \mathcal E,\mathcal D_B \right] \, \right\|_2^2  \label{CGP_commutator}
          \end{align}
  \end{enumerate}
\end{prop}
\begin{proof}
    \textbf{(i)} Equation \eqref{CGP_formula} can be equivalently written as
    \begin{align*}
    &C_{2,B}(\mathcal E) = \int d \mu_{unif} (\bs p)\,  \left[ \braket{\mathcal Q_B \mathcal E \textstyle \sum_i \displaystyle p_i P_i , \mathcal Q_B \mathcal E \textstyle \sum_j \displaystyle p_j P_j} \right] \\
    &= \sum_{i,j} \int d \mu_{unif} (\bs p)\, \left[p_i p_j\right] \braket{\mathcal Q_B \mathcal E  P_i , \mathcal Q_B \mathcal E P_j} \\
    &= \sum_{i,j} \int d \mu_{unif} (\bs p)\, \left[p_i p_j\right] \Big( \braket{\mathcal E  P_i , \mathcal E P_j}  - \braket{\mathcal D_B \mathcal E  P_i , \mathcal D_B \mathcal E P_j} \Big) \,\;,
    \end{align*}
    where the last equality was obtained using the definition $\mathcal Q_B = I - \mathcal D_B$ and the fact that $\mathcal D_B$  is a hermitian orthogonal projector (and therefore idempotent).
    Assuming that the quantum processes is unital ($\mathcal E (I) = I$) and using the fact that for uniform input ensemble of states $\int d \mu_{unif} (\bs p)\, \left[p_i p_j\right] = [d(d+1)]^{-1}(1+\delta_{ij})$ (see, e.g., \cite{coherence_2} for a derivation) the result follows. \\ \\
    \textbf{(ii)} We have
    \begin{gather*}
    \left\| \,\left[ \mathcal E,\mathcal D_B \right] \, \right\|_2^2 = \Tr\left( \mathcal E ^\dagger \mathcal E \mathcal D_B \right) + \Tr\left( \mathcal E \mathcal E^\dagger  \mathcal D_B \right) -2 \Tr \left( \mathcal D_B \mathcal E ^\dagger \mathcal D_B \mathcal E \right) \,.
    \end{gather*}
    From the normality assumption it follows that the first two terms are equal. The superoperator traces can then be evaluated using the Hilbert-Schmidt operator inner product as $\Tr \left( \mathcal X \right) = \sum_{i,j} \braket{ \,\ket{i}\bra{j} , \mathcal X \left( \ket{i} \bra{j} \right) }$ which yields
    $ \Tr\left( \mathcal E ^\dagger \mathcal E \mathcal D_B \right)= \sum_i \braket{\mathcal E P_i ,\mathcal E P_i  }$. A similar calculation for the remaining term gives $\Tr \left( \mathcal D_B \mathcal E ^\dagger \mathcal D_B \mathcal E \right) = \sum_i \braket{\mathcal D_B\mathcal E P_i ,\mathcal D_B\mathcal E P_i}$ and hence the result follows.
\end{proof}

Let us now make a couple of remarks for CGP based on the Hilbert-Schmidt 2-norm.
The original averaging definition for the CGP, Equation \eqref{CGP_formula}, surprisingly admits in this case the much simpler form given by Eq.~\eqref{CGP_commutator}. The last equation also implies that the 2-norm CGP for unital quantum channels defined originally is nothing more than \textit{a measure of the degree of non-commutativity between $\mathcal E$ and the dephasing channel $\mathcal D_B$} (see also \cite{zanardi2017quantum}).


\subsubsection{Relative entropy based monotone}

A coherence monotone for MIO (and therefore all subclasses of free operations, see e.g.~\cite{chitambar2016comparison}) is obtained using relative entropy \cite{baumgratz2014quantifying}:
\begin{align}
  c_{r,B} (\rho) \coloneqq \min_{\sigma \in I_B} S(\rho \,\|\, \sigma)  =   S (\mathcal D_B \rho) - S(\rho) \label{relative_entropy}
\end{align}
Let $C_{r,B}(\mathcal E)$ be the CGP Eq.~\eqref{CGP_formula} of the quantum channel $\mathcal E$ with $c_B = c_{r,B}$. Then, from Eq.~\eqref{relative_entropy} it is immediate that
\begin{align}
  C_{r,B}(\mathcal E) = \int d \mu_{unif} (\bs p)\,  \left[S(\mathcal D_B \mathcal E \rho _{in}(\bs p))  - S(\mathcal E \rho _{in}(\bs p)) \right]  \,\;. \label{CGP_relative}
\end{align}
Under the class of IO the monotone $c_r(\rho)$ has an operational interpretation as the optimal rate of asymptotic coherence distillation, i.e., $ c_{r,B} (\rho) = \sup R $ such that $\rho^{\otimes n} \xmapsto{IO} \Phi_2^{\otimes n R} $ as $n\to \infty$, where $\Phi_2$ is the maximal coherence qubit state \cite{winter2016operational}.
The relative entropy CGP of $\mathcal E$ therefore admits an operational interpretation as the \textit{average rate of distillable coherence} (under IO) of the processed state $\mathcal E \left(\rho_{in} \right)$.

Zhang \textit{et al.} in \cite{ZhangToAppear} have obtained explicit expressions for the relative entropy CGP, Eq.~\eqref{CGP_relative}, when $\mathcal E$ is a unitary channel. In this article we obtain expressions for dephasing channels, as will be shown in \autoref{relative_section}.

\section{Coherence generating power of maximally and partially dephasing processes} \label{main_section}

\subsection{Maximally and partially dephasing processes} \label{maximal_partial_definitions_section}

We begin by distinguishing between two families of dephasing processes.
\begin{definition}
   We characterize a dephasing map $\mathcal D_B: \mathcal S (\mathcal H) \to \mathcal S(\mathcal H)$
   as a \textbf{maximally} dephasing channel iff $\Rank \left[ \Pi_i \right] = 1$ for all $\Pi_i \in B$. Otherwise the map is a \textbf{partially} dephasing channel.
\end{definition}
In other words, for a \textit{maximally} dephasing process there exists an orthonormal basis such that all output states of the channel are diagonal with respect to that basis. If the output states are block diagonal (with non-trivial blocks) we call the dephasing \textit{partial}. Notice, however, that the basis over which this is achieved is not uniquely specified: for example, any permutation or phase shift on the basis elements will still preserve the diagonal outcomes. Nevertheless, the complete set of projectors $B$ completely characterizes the (maximally or partially) dephasing channel:
\begin{prop} \label{unique_dephasing_label}
  The set of projectors $B$  corresponding to the (maximally or partially) dephasing channel $\mathcal D_B$ is unique.
\end{prop}
\begin{proof}
  Clearly, the set of projectors $B = \{ \Pi_i \}$ is a set of Kraus operators for the CPTP map $\mathcal D_B$. We need to show that if $B'=\{\Pi'_i\}_i$ is a complete family of orthogonal projectors  with $\mathcal D_B = \mathcal D_{B'}$, then $B=B'$. Indeed, the two Kraus decompositions describe the same channel iff there exists a unitary $U$ such that $\Pi'_i = \sum_j U_{ij} \Pi_j$ (see, e.g. \cite{wolf2012notes}). But $\sum_i \Pi'_i = I$ which is true only if $\sum_i U_{ij} = 1 \, \forall j$. On the other hand, $U$ is a unitary matrix so the columns form orthonormal vectors. The last two properties can hold together only when $U$ is a permutation matrix. But a permutation of the Kraus operator indices doesn't affect the set $B$, therefore $B=B'$.
\end{proof}

Dephasing processes can be viewed, for example, as non-selective orthogonal measurements. In the case of maximally dephasing the measured observable is non-degenerate and thus all projectors $B'=\{P'_i\}_{i=1}^d$ of $\mathcal D_{B'}$ are rank-1, corresponding to the distinct eigenvalues of the measured observable. On the other hand, in the case of a degenerate observable non-trivial subspaces $B'=\{\Pi'_i\}$ occur with $\Rank \left( \Pi_i \right)$ equal to the degeneracy of the $i$-th eigenvalue, i.e., the dephasing is partial.

\subsection{Coherence generating power of maximally dephasing channels} \label{maximal_section}

Let us begin with a simple remark. For any maximally dephasing process there exists a basis $\{ \ket{i}\}_{i=1}^d$ such that the ``coherences'' of the output state (i.e., elements $\bra{i} \mathcal D_B \left( \rho \right)  \ket{j}$ for $i \ne j$) vanish. Nevertheless, this does not necessarily imply an incoherent output state since no statement has been made regarding the reference basis with respect to which coherence is quantified.  Consider for example the qubit case, maximally dephasing $\mathcal D _{B'}$ where $B' = \{ \ket{+} \bra{+}, \ket{-} \bra{-} \}$, and coherence quantified with respect to $B = \{ \ket{0} \bra{0} , \ket{1} \bra{1} \}$. Then clearly $\Ran(\mathcal D_B) \ne \Ran(\mathcal D_{B'})$: not all $B'$ dephased states are incoherent in $B$.

In what follows we expand on this simple observation, quantifying through Eq.~\eqref{CGP_formula} how ``valuable'' a maximally dephasing channel is at creating coherence or, in other words, we calculate how much resource a maximally dephasing channel generates on average after acting on incoherent states. The 2-norm and relative entropy coherence quantifiers are adopted, each relevant in a different class of free operations of the theory.

\subsubsection{Hilbert-Schmidt 2-norm coherence}

\begin{prop}[2-norm CGP of maximally dephasing] \label{2norm_CGP_prop}
  Let $B = \{ \ket{i} \bra{i}\}_{i=1}^d$ and $B' = \{ \ket{i'} \bra{i'} \}_{i=1}^d$ be complete families of rank-1 orthogonal projectors and $U \in U(d)$ be a unitary operator such that $\ket{i'} = U \ket{i}$ for all $i=1,\dotsc,d$. Then
  \begin{enumerate}[(i)]
    \item The 2-norm CGP of the maximally dephasing channel $\mathcal D_{B'}$ is given by
    \begin{gather}
      C_{2,B}(\mathcal D_{B'}) = \frac{1}{d(d+1)} \Tr\left[  X_U X_U^T \left( I - X_U X_U^T \right) \right]\label{dephasing_2norm_CGP}
    \end{gather}
    where $X_U \in \mathbb R^{d \times d}$ is bistochastic with $\left(X_U\right)_{ij} = |\!\braket{i |U| j}\!|^2$.
    \item Alternatively, on the superoperator level,
          \begin{gather}
            C_{2,B}(\mathcal D_{B'}) = \frac{1}{2d(d+1)} \left\| \, \left[ \mathcal D_{B'},\mathcal D_{B} \right]  \, \right\|_2^2 \label{dephasing_commutator_CGP}
          \end{gather}
    \item
    \begin{gather}
    0 \le  C_{2,B}(\mathcal D_{B'}) \le C_{2,B}^{max} (d) \coloneqq \frac{d-1}{4d(d+1)} \,\; , \label{2-norm_CGP_bounds}
    \end{gather}
    where the lower bound is achieved iff $\left[ \mathcal D_B, \mathcal D_{B'} \right] = 0$.
  \end{enumerate}
\end{prop}
\begin{proof}
    \textbf{(i)} First we notice that $\mathcal D_{B'} = \mathcal U \mathcal D_B \mathcal U^\dagger$, where $\mathcal U (\cdot) = U (\cdot) U^\dagger$.
    Next we calculate the quantities $\sum_i \braket{\mathcal D_{B'} P_i,\mathcal D_{B'} P_i}$ and $\sum_i \braket{\mathcal D_{B} \mathcal D_{B'} P_i, \mathcal D_{B} \mathcal D_{B'} P_i}$, appearing in  equation \eqref{CGP_unitals}. A straightforward calculation gives
    \begin{gather*}
    \sum _i \braket{\mathcal D_{B'} P_i,\mathcal D_{B'} P_i} = \sum_i \Tr \left( P_i \mathcal U \mathcal D_B \mathcal U^\dagger P_i \right)  = \Tr \left[  X_U X_U^T \right] \,\;.
    \end{gather*}
    A similar calculation for the other term gives
    \begin{align*}
    &\sum_i \braket{\mathcal D_{B} \mathcal D_{B'} P_i, \mathcal D_{B} \mathcal D_{B'} P_i} = \sum_i \Tr \left( P_i \mathcal D_{B'} \mathcal D_{B} \mathcal D_{B'}  P_i \right)\\
    = & \sum_i \Tr \left( P_i \mathcal U \mathcal D_{B}  \mathcal U^\dagger  \mathcal D_{B}  \mathcal U  \mathcal D_{B} \mathcal U^\dagger P_i\right) = \Tr\left[\left(X_U X_U^T\right)^2\right] \,\;.
    \end{align*}
    The result for $C_{2,B}\left( \mathcal D_{B'} \right)$ follows. The bistochasticity of the matrix $\left(X_U\right)_{ij}$ is a direct consequence of the unitarity of $U$.\\ \\
    \textbf{(ii)} We have $\mathcal D_B = \mathcal D_{B} ^\dagger$, i.e., the maximally dephasing channels are hermitian therefore normal, and also unital. The claim hence follows by setting $\mathcal E =  \mathcal D_{B'}$ in equation \eqref{CGP_commutator}. \\
    \textbf{(iii)} The lower bound properties follow immediately from equation \eqref{dephasing_commutator_CGP}. For the upper bound observe that the matrix $X_U X_U^T$ is positive semi-definite and also bistochastic (as product of bistochastic matrices). Now, since bistochastic matrices have at least one eigenvalue equal to one, the difference $\left( \Tr\left[  X_U X_U^T \right] - \Tr\left[\left(X_U X_U^T\right)^2\right] \right)$ is bounded from above by $(d-1)[1/2 - (1/2)^2] = (d-1)/4$. The numerical factor of $1/2$ corresponds to the ($(d-1)$-fold degenerate) eigenvalue of $X_U X_U^T$ which maximizes the difference (since $0 \le \lambda \le 1$ for all eigenvalues). Notice that it is not \textit{a priori} guaranteed that a unitary matrix $U$ (corresponding to such an $X_U$) exists. We tackle this question in \autoref{appendix_maximum} of the Appendix.
\end{proof}

\begin{figure}[t]
  \centering
  \includegraphics[width= \columnwidth]{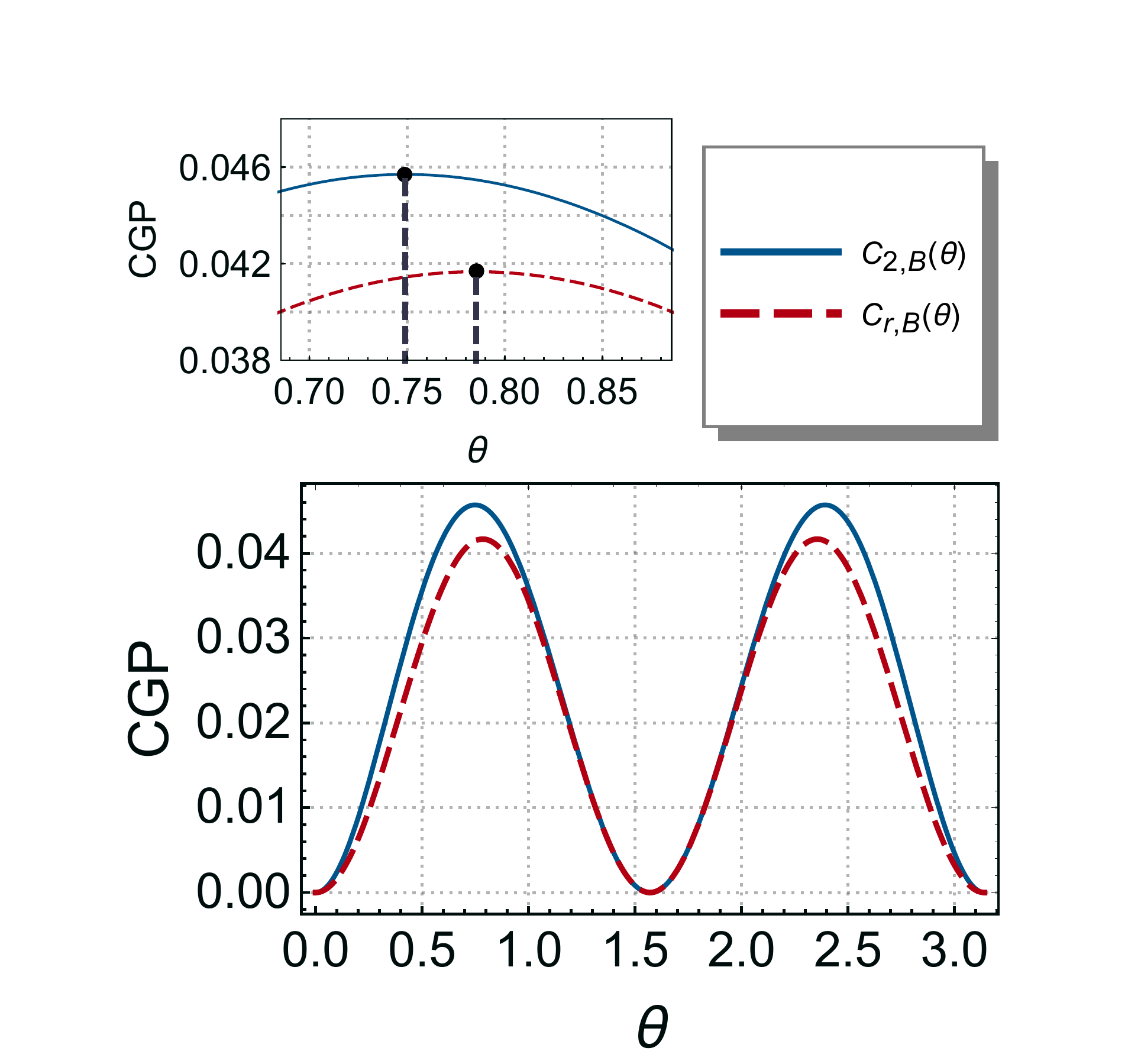}
  \caption{Single qubit coherence generating power of maximally dephasing channels as a function of the angle $\theta$ (see Ex.~\autoref{qubit_example_1} \& Ex.~\autoref{qubit_example_2}), for the relative entropy $c_{r,B}$ and 2-norm $c_{2,B}$ coherence quantifiers. Notice that the maximum is obtained for slightly different values of the angle $\theta$.} \label{relative_figure}
\end{figure}

Let us now specialize the above for a 2-level system.

\begin{example}[Single qubit maximal dephasing: 2-norm] \label{qubit_example_1}
  Consider a qubit ($\mathcal H = \Span \{ \ket{0},\ket{1} \}$) with its coherence quantified with respect to $B = \{ \ket{0} \bra{0}, \ket{1} \bra{1} \}$ and any maximally dephasing channel $\mathcal D_{B'}$, where $B' = \{ \ket{\psi_0} \bra{\psi_0}, \ket{\psi_1} \bra{\psi_1} \}$ ($\braket{\psi_0|\psi_1}=0$). For the qubit case $\mathcal D_{B'}$  can be parametrized through the (Bloch sphere) angles $\theta$ and $\phi$, where $\ket{\psi_0(\theta,\phi)} \bra{\psi_0(\theta,\phi)} = 1/2 \left( I + \bs v \cdot \bs \sigma \right)$, with $\bs v = (\sin \theta \cos \phi, \sin \theta \sin \phi, \cos \theta)$. From Prop.~\autoref{2norm_CGP_prop} we have
  \begin{gather*}
      X_U = \begin{pmatrix}
  \cos^2(\theta/2) & \sin^2(\theta/2) \\
  \sin^2(\theta/2) & \cos^2(\theta/2)
  \end{pmatrix}
  \end{gather*}
  (independent of $\phi$) and hence $C_{2,B} = 1/24 \sin^2(2 \theta)$ (with $\theta \in [0,\pi])$. Observe that the upper bound from Eq.~\eqref{2-norm_CGP_bounds} is achieved for $\theta \in \{ \pi / 4,3 \pi /4 \}$. On the other hand, for $\theta \in \{ 0, \pi/2, \pi \}$ the CGP vanishes. The cases $\theta = 0$ and $\theta = \pi$ give $B = B'$ but the case $\theta = \pi/2$ (for all $\phi$) corresponds to $B'$ being a mutually unbiased basis of $B$. In all such cases $\left[ \mathcal D_{B},\mathcal D_{B'} \right] = 0$.
\end{example}

In Ref.~\cite{coherence_1} it was shown that, for unital quantum channels $\mathcal E$,
\begin{gather}
C_{2,B} (\mathcal E) \le \frac{d-1}{d(d+1)} = 4 C_{2,B}^{max}(d) \,\;,
\end{gather}
where the maximum is achieved over unitary $\mathcal U (\cdot) = U (\cdot) U^\dagger $ with $\left| \braket{i | U | j} \right|^2 = 1/d$ for all $i,j$ (e.g. the quantum Fourier transform). For maximally dephasing processes observe that the optimal 2-norm CGP can be at most one quarter of the maximum value for the CGP (over all unital CPTP maps). In \autoref{appendix_maximum} of the Appendix we provide an explicit construction to show that the upper bound for maximally dephasing processes $C_{2,B}^{max}(d)$ is achievable for Hilbert space dimensions $d \le 13$ while we also give a set of sufficient conditions for the bound to be achievable in any dimension.

In Example \autoref{qubit_example_1} it was noticed that for a qubit system $C_{2,B}\left(  \mathcal D_{B'} \right)$ vanishes when $B'$ is a mutually unbiased basis of $B$. This observation holds true for any finite dimensional system. Consider a mutually unbiased basis $\{ \ket{i'}\}_{i=1}^d$ such that $\left| \braket{j|i'} \right|^2 = 1/d$. Then the matrix $X_U$ from Eq.~\eqref{dephasing_2norm_CGP} has matrix elements $\left( X_U \right)_{ij} = 1/d$ and therefore $C_{2,B}\left(  \mathcal D_{B'} \right) = 0$.

\subsubsection{Relative entropy of coherence} \label{relative_section}

We will now give a set of results ``parallel'' to Prop.~\autoref{2norm_CGP_prop} using as coherence quantifier, instead of $c_{2,B}$, the relative entropy of coherence $c_{r,B}$. Given a $d$-dimensional probability vector $\bs p$ we denote its \textit{subentropy} \cite{jozsa1994lower} as $Q(\bs p)$, defined as
\begin{gather}
  Q(\bs p) \coloneqq -\sum_{i=1}^d \frac{p_i^d \ln p_i}{\prod_{j \ne i} (p_i - p_j)} \,\;,
\end{gather}
where if two or more of the $p_i$'s are equal then the limit should be taken as they become equal. The definition is extended to column-stochastic matrices $X \in  (\mathbb R_0^+)^{d \times d}$ as $Q(X) \coloneqq 1/d \sum_j Q(\bs p_j) $, where $(\bs p_j)_{i} = (X)_{ij}$.

\begin{prop}[Relative entropy CGP of maximal dephasing]   \label{relative_CGP_prop}
  Let $B = \{ \ket{i} \bra{i}\}_{i=1}^d$ and $B' = \{ \ket{i'} \bra{i'} \}_{i=1}^d$ be complete families of rank-1 orthogonal projectors and $U \in U(d)$ be a unitary operator such that $\ket{i'} = U \ket{i}$ for all $i=1,\dotsc,d$. Then
  \begin{enumerate}[(i)]
    \item The relative entropy CGP of the maximally dephasing channel $\mathcal D_{B'}$ is given by
    \begin{gather}
      C_{r,B}(\mathcal D_{B'}) = Q\left(X_U X_U^T \right) - Q\left(X_U \right) \label{relative_entropy_CGP_formula}
    \end{gather}
    where $X_U \in \mathbb R ^{d \times d}$ is bistochastic  with $\left(X_U\right)_{ij} = |\!\braket{i |U| j}\!|^2$ and $Q(X)$ denotes the subentropy of $X$.
    \item $C_{r,B}(\mathcal D_{B'}) = 0$ if and only if $\left[ \mathcal D_B , \mathcal D_{B'} \right] = 0$.
  \end{enumerate}
\end{prop}
\begin{proof}
  \textbf{(i)} The proof is based on a lemma from \cite{ZhangToAppear}, stating that
  \begin{gather} \label{subentropy_lemma}
  \int d \mu_{unif} (\bs p) \, H\left(  \sum_j B_{ij} p_j \right) = H_d -1 +Q\left(B^T\right) \,\;,
  \end{gather}
  where $B$ is a $d \times d$ bistochastic matrix and $H_d$ is the $d$-th harmonic number. Here we want to calculate the quantity in Eq.~\eqref{CGP_relative} for $\mathcal E = \mathcal D_{B'} = \mathcal U \mathcal D_{B} \mathcal U ^\dagger$, where $\mathcal U (\cdot) = U (\cdot) U^\dagger$. Observe that
  \begin{align*}
  S(\mathcal D_{B'} \rho_{in}) &= S\left(\mathcal U \mathcal D_B \mathcal  U^\dagger \sum_i(p_i P_i)\right) = S\left( \mathcal D_B \mathcal  U^\dagger \sum_i(p_i P_i) \right)\\ & = H\left(\sum_j (X_{U^\dagger})_{ij} p_j\right) = H\left(\sum_j (X_U)^T_{ij} p_j\right) \,\;,
  \end{align*}
  where we used the unitary invariance of the von Neumann entropy and the fact that $\mathcal D_B \mathcal U \sum_i (p_i P_i) = \sum_{i,j} \left(X_U\right)_{ij} p_j P_j$. From the lemma Eq.~\eqref{subentropy_lemma} we therefore get
  \begin{gather*}
  \int d \mu_{unif} (\bs p) \, S(\mathcal D_{B'} \rho_{in}) = H_d -1 +Q(X_U)\,\;.
  \end{gather*}
  In a similar fashion,
  \begin{align*}
    S(\mathcal D_{B}\mathcal D_{B'} \rho_{in}) & = S\left(\mathcal D_{B} \mathcal U \mathcal D_{B} \mathcal U^\dagger \sum_i(p_i P_i) \right) \\
    &= S\left( \sum_{i,j} (X_U X_U^T)_{ij} p_j P_j  \right)\,\;,
  \end{align*}
  therefore
  \begin{gather*}
  \int d \mu_{unif} (\bs p) \, S(\mathcal D_{B} \mathcal D_{B'} \rho_{in}) = H_d -1 +Q(X_U X_U^T) \,\;.
  \end{gather*}
  Combining the two expressions we get the desired result. \\ \\
  \textbf{(ii)} The relative entropy of coherence is a faithful measure ($c_{r,B}(\sigma) = 0 $ iff $\sigma \in I_B$) and, therefore, $C_{r,B}\left( \mathcal D_{B'} \right) = 0$ iff $\mathcal D_{B'} (\sigma) \in I_B $ for all $\sigma \in I_B$. In other words, $C_{r,B}\left( \mathcal D_{B'} \right) = 0$ is equivalent to $\Ran \left( \mathcal D_B \right)$ being invariant under the action of $\mathcal D_{B'}$. But $\mathcal D_{B'}$ is hermitian and therefore normal and as a result the vanishing CGP condition holds iff both $\Ran \left( \mathcal D_B \right)$ and $\Ker \left( D_B \right)$ are invariant over the action of $\mathcal D_{B'}$. We will now argue that this condition is equivalent to $\left[ \mathcal D_B , \mathcal D_{B'} \right] = 0$.

  Suppose $\left[ \mathcal D_B , \mathcal D_{B'} \right] = 0$. Then clearly $\Ker \left(\mathcal D_B \right)$ is an invariant subspace of $\mathcal D_{B'}$ and so is therefore $\Ran \left(\mathcal D_B \right)$ (from normality of $\mathcal D_{B'}$).

  Conversely, suppose that both $\Ker \left( D_B \right)$ and  $\Ran \left( \mathcal D_B \right)$ are invariant over the action of $\mathcal D_{B'}$. Then for any operator $X \in \Ker \left(\mathcal D_B \right)$ we have $\mathcal D_{B'} \mathcal D_B X = \mathcal D_{B} \mathcal D_{B'}X = 0$. But also for any $Y \in \Ran \left( \mathcal D_B \right)$ we have $\mathcal D_{B'} \mathcal D_B Y = \mathcal D_{B} \mathcal D_{B'} Y$ (since $\mathcal D_B Y = Y$). As a result $\left[ \mathcal D_B , \mathcal D_{B'} \right] = 0$.

  Notice that the proof relies only on the faithfulness of the coherence measure $c_{r,B}$ and thus the result holds true for any faithful coherence measure $c_{B}$.
\end{proof}

\begin{example}[Single qubit maximal dephasing: relative entropy] \label{qubit_example_2}
  Assume a single qubit as in Ex.~\autoref{qubit_example_1}. Using Eq.~\eqref{relative_entropy_CGP_formula} we can calculate the relative entropy CGP for maximally dephasing $\mathcal D_{B'}$. Setting $c \coloneqq \cos^2 (\theta/2)$ and $s \coloneqq \sin^2 (\theta/2)$, we get
  \begin{align}
  C_{r,B}\left( \mathcal D_{B'}  \right) = &\frac{1}{(c-s)^{2}} \Big( \big[ c^2(c-s) \log c + 2 s^2 c^2 \log(2cs) \nonumber \\ &- \frac{1}{2}(c^2 + s^2)^2 \log(c^2 + s^2) \big] + \left[ s \leftrightarrow c \right] \Big) \,\;.
  \end{align}
  The resulting function is compared with the corresponding one from Ex.~\autoref{qubit_example_1} in \autoref{relative_figure}.
\end{example}

Notice that equations \eqref{dephasing_2norm_CGP} and \eqref{relative_entropy_CGP_formula} are functions of the maximally dephasing superoperators $\mathcal D_B$ and $\mathcal D_{B'}$, which are characterized (according to Prop.~\autoref{unique_dephasing_label}) by the sets of rank-1 projectors $B$ and $B'$, respectively. It is therefore expected that any choice of a unitary $U$ in those equations such that $U \ket{i} = \e^{\im \theta_i} \ket{\sigma(i)'}$, with $\theta_i \in \mathbb R$ and $\sigma \in \mathcal S_d$ (permutation) leaves both $C_{2,B}(\mathcal D_{B'})$ and $C_{r,B}(\mathcal D_{B'})$ unaffected.

\subsection{Coherence generating power of partially dephasing channels} \label{partial_section}

In this section we extend the previous results for the 2-norm CGP to partially dephasing channels.

\begin{prop}[2-norm CGP of partial dephasing]
  Let $B = \{ P_i = \ket{i} \bra{i}\}_{i=1}^d$ and $B' = \{ \Pi_k =  \sum_{\alpha=1}^{d_k} \ket{k,\alpha} \bra{k,\alpha} \}_{k=1}^{r}$ be complete families of orthogonal projectors, where $\{ \ket{k,\alpha} \}_{k,\alpha}$ is an orthonormal basis and $d_k \coloneqq \Rank \left( \Pi_k \right)$. Then
  \begin{gather}
    C_{2,B} \left( \mathcal D_{B'} \right) = \frac{1}{d(d+1)}  \Tr \bigl[ Z  \left(I - Z \right) \bigr]\,\;, \label{2-norm_CGP_partial}
  \end{gather}
  where $Z \in \mathbb R ^{d \times d}$ with
  \begin{align}
    Z_{ij} &= \Tr \left( \sum_k P_i \Pi_k P_j \Pi_k \right)  \\ &= \sum_{k=1}^r \sum_{\alpha,\beta=1}^{d_k} \braket{i|k,\alpha} \braket{k,\alpha|j} \braket{j|k,\beta} \braket{k,\beta|i} \,\;. \label{Z_matrix_definition}
  \end{align}
\end{prop}
\begin{proof}
    Our starting point is Eq.~\eqref{CGP_unitals} with $\mathcal E (\cdot)= \mathcal D_{B'} (\cdot) = \sum_k \Pi_k (\cdot) \Pi_k$. We have
    \begin{align*}
    &\sum_i \braket{\mathcal D_{B'} P_i, \mathcal D_{B'} P_i } = \sum_i \braket{\mathcal D_{B'} P_i, P_i }\\
     = &\sum_{i,k} \Tr \left( \Pi_k P_i \Pi_k P_i \right) = \Tr \left( Z \right)\,\;.
     \end{align*}
     The more detailed expression for $Z$ in terms of the basis elements $\{ \ket{k,\alpha} \}$ follows by writing $\Pi_k = \sum_{\alpha=1}^{d_k} \ket{k,\alpha} \bra{k,\alpha}$ and performing the trace. The second term becomes
     \begin{align*}
     &\sum_i \braket{\mathcal D_B \mathcal D_{B'} P_i, \mathcal D_B \mathcal D_{B'} P_i } =  \sum_i \braket{ \mathcal D_{B'} \mathcal D_B \mathcal D_{B'} P_i, P_i }  \\
     = &\sum_{i,j,k,l} \Tr \left( \Pi_l P_j \Pi_k P_i \Pi_k P_j \Pi_l P_i \right) \\
     = &\sum_{i,j} \Tr \left[   \left( \sum_l \Pi_l P_i \Pi_l P_j \right)  \left( \sum_k \Pi_k P_i \Pi_k P_j \right) \right] = \Tr \left( Z ^2\right)\,\;.
     \end{align*}

     By the cyclic property of the trace, the matrix $Z$ is symmetric and also real since $Z_{ij}^* = \Tr \left[ \sum_k \left( P_i \Pi_k P_j \Pi) \right)^\dagger \right] = Z_{ij}$.
\end{proof}

Let us observe that the $Z$ matrix from Eq.~\eqref{Z_matrix_definition} can be also expressed as a function of the unitary $U$  connecting $B$ and $B'$ (as in Prop.~\autoref{2norm_CGP_prop}). Indeed, for any $U$ such that $ U \ket{j} = \ket{k(j),\alpha(j)} $ then $\braket{i| k(j),\alpha(j)} = U_{ij}$. Of course, such a $U$ is far from unique since it depends on the choice of basis for each subspace corresponding to a $\Pi_l $. In the case where the dephasing is maximal, rather that partial, Eq.~\eqref{2-norm_CGP_partial} reduces to Eq.~\eqref{dephasing_2norm_CGP}. Indeed, since $d_k = 1$ for all $k=1, \dotsc, d$, the label $\alpha$ in $\ket{k,\alpha}$ becomes redundant ($\alpha = 1$) and we can use the notation $\ket{k,\alpha} \to \ket{k'}$. As a result $Z = X_U X_U^T$, where $U$ is a unitary such that $\ket{i'} = U \ket{i}$.

\begin{figure}[t]
  \centering
  \includegraphics[width= \columnwidth]{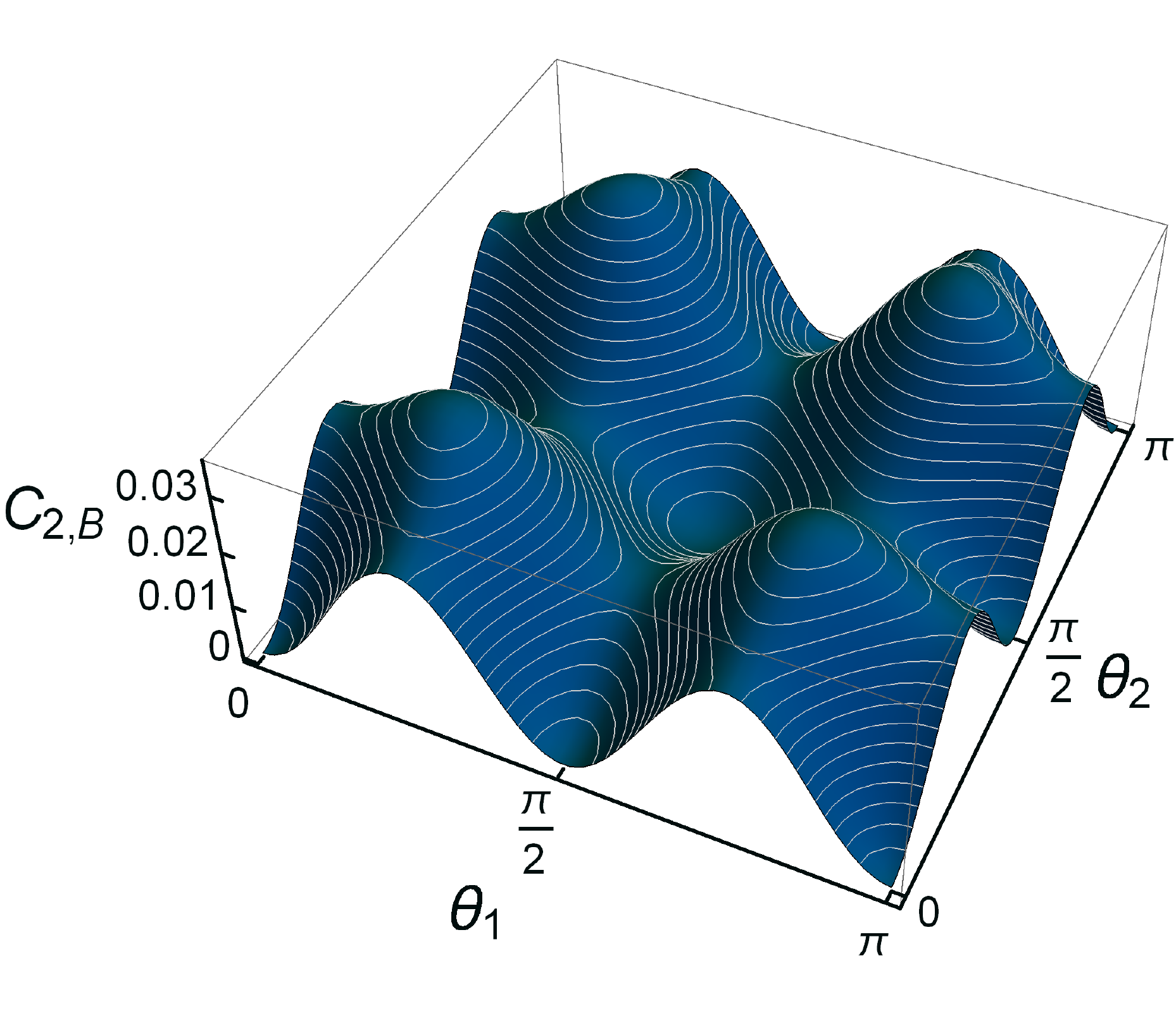}
  \caption{$C_{2,B}\left(  \mathcal D _{B'}\right)$ of Ex.~\autoref{CGP_example_2} is a (symmetric) function of both $\theta_1$ and $\theta_2$. The difference $C_{2,B}\left(  \mathcal D _{B'}\right) - C_{2,B}\left(  \mathcal D _{B''}\right)$ can be positive or negative, depending on the values of $\theta_1$ and $\theta_2$. The maximum value $M_m$ of $C_{2,B}\left(  \mathcal D _{B'}\right)$ satisfies $M_p < M_m < C_{2,B}^{max}(d=4)$ demonstrating that for a 2-qubit system the maximally dephasing $\mathcal D_{B'}$ such that $C_{2,B}(\mathcal D_{B'}) = C_{2,B}^{max}(d=4)$ is not of the form $\mathcal D_{B'} = \mathcal D_{B'_1} \otimes \mathcal D_{B'_2}$ (for any qubit bases $B'_1,B'_2$).}
  \label{figure_1}
\end{figure}

\begin{example}[2-qubit partial and maximal dephasing] \label{CGP_example_2}
  Consider a 2-qubit Hilbert space $\mathcal H \cong \mathbb C^2 \otimes \mathbb C^2$ with $\mathcal H = \Span \{ \ket{00}, \ket{01}, \ket{10}, \ket{11} \}$ and coherence quantified with respect to $B = \{ P_{00},P_{01},P_{10},P_{11} \}$ ($P_{ij} = \ket{ij}\bra{ij}$). Define the complete set of rank-1 projectors $B' = \{ P'_{00},P'_{01},P'_{10},P'_{11} \}$ where $P'_{ij} = \ket{\psi_i (\theta_1, \phi_1) \psi_j(\theta_2, \phi_2)} \bra{\psi_i (\theta_1, \phi_1) \psi_j(\theta_2, \phi_2)}$ (notation as in Example \autoref{qubit_example_1}). We look at dephasing $\mathcal D_{B''}$ with respect to $B'' = \{ \Pi_1 = P'_{00} + P'_{01}, \Pi_2 = P'_{10} + P'_{11} \}$. Applying Eq.~\eqref{2-norm_CGP_partial} we get $C_{2,B}\left(  \mathcal D _{B''}\right) = \frac{1}{40} \sin^2(2 \theta_1)$. Clearly, the maximum value over the dephasing channels examined is $M_p = 1/40$. The CGP of the maximally dephasing  $C_{2,B}\left(  \mathcal D _{B'}\right)$ is plotted in \autoref{figure_1}.
\end{example}

\section{Dephasing processes through Lindblad dynamics} \label{Lindblad_section}

\subsection{Coherence generating power of dephasing Lindbladians} \label{CGP_Lindblad_section}

In this section we consider quantum dynamical processes that lead in the long time limit to dephasing of the system under consideration. More specifically, we assume Markovian dynamics of Lindblad form (see, e.g., \cite{breuer2002theory})
\begin{gather}
  \dot \rho = \mathcal L \, \rho \coloneqq - \im [ H,\rho ] + \sum_\alpha \left( L_\alpha \rho L_\alpha ^\dagger - \frac{1}{2} \left\{ L_\alpha ^\dagger L_\alpha , \rho \right\} \right) \,\;, \label{Lindblad}
\end{gather}
where $H$ is the system Hamiltonian and $\{ L_\alpha \}_\alpha$ are the Lindblad operators. We first distinguish those Lindblad time evolutions that maximally dephase in the long time limit.
\begin{definition}
  We characterize an operator $\mathcal L \in \mathcal B\left(\mathcal B \left(\mathcal H \right) \right)$ of the Lindblad form Eq.~\eqref{Lindblad} as a \textbf{maximally dephasing Lindbladian} iff $\lim_{t \to \infty} \exp\left( \mathcal L t \right) = \mathcal D_B$, for some maximally dephasing channel $\mathcal D_B$.
\end{definition}

Our next steps will be to characterize the maximally dephasing Lindbladians and then calculate the 2-norm CGP for all such time evolutions as a function of time. Naturally, we will recover part of our previous results for the CGP of maximally dephasing channels in the limit $t \to \infty$.

We begin with a Lemma.

\begin{lemma*}
  Let $\mathcal L$ be a Lindbladian of the general form Eq.~\eqref{Lindblad}. Then $\ket{\psi} \bra{\psi} \in \Ker \left( \mathcal L \right)$ if and only if the following conditions hold simultaneously: \begin{inparaenum}[(a)]
  \item $\ket{\psi}$ is an eigenvector of $L_\alpha$ for all $\alpha$, and
  \item $\ket{\psi}$ is an eigenvector of $\,\im H + \frac{1}{2}\sum_\alpha \braket{\psi | L_\alpha | \psi} L_\alpha^\dagger$.
  \end{inparaenum}
\end{lemma*}
\begin{proof}
  Let $\{  \ket{\psi _j ^\perp}\}_{j=1}^{d-1}$ be a set of orthonormal vectors with $\braket{\psi | \psi_j^\perp} = 0$ for all $j = 1, \dotsc, {d-1}$. We have $\mathcal L \left( \ket{\psi} \bra{\psi} \right) = 0 $ iff the following hold true:
  \begin{inparaenum}[$(a')$]
  \item $\braket{\psi  |\mathcal L \left( \ket{\psi} \bra{\psi} \right)| \psi } = 0$,
  \item $\braket{\psi  |\mathcal L \left( \ket{\psi} \bra{\psi} \right)| \psi_j ^\perp } = 0$ for all $j$,
  \item $\braket{\psi_k ^\perp  |\mathcal L \left( \ket{\psi} \bra{\psi} \right)| \psi_j ^\perp } = 0$ for all $j,k$.
  \end{inparaenum}
  By plugging in the Lindblad form for $\mathcal L$ (equation \eqref{Lindblad}), it follows that condition $(a')$ holds true iff condition $(a)$ it true. Then, given $(a)$, condition $(b')$ is trivially satisfied. Finally condition $(c')$, again given $(a)$, reduces to $(b)$.
\end{proof}
We now establish the necessary and sufficient conditions to have maximal dephasing under Lindbladian dynamics and then we calculate the 2-norm CGP for all such processes.
\begin{prop} \label{maximal_dephasing_lindblad_prop}
  Let $B' = \{ P'_i \coloneqq \ket{i'} \bra{i'} \}_{i=1}^d$ and $\mathcal D_{B'}$ be the associated maximally dephasing channel. Then for Lindbladian dynamics:
  \begin{enumerate}[(i)]
    \item $\lim_{t \to \infty} \exp \left( \mathcal L \, t  \right) = \mathcal D_{B'}$
          if and only if the following conditions hold simultaneously: \begin{inparaenum}[(a)]
                            \item The Hamiltonian $H$ is diagonal in $B'$.
                            \item All Lindblad operators $L_\alpha$ are diagonal in $B'$.
                            \item For every $i \ne j$ there exists an $\alpha$ such that $\braket{ i' | L_\alpha | i'} \ne \braket{ j' | L_\alpha | j'}$.
                        \end{inparaenum}
    \item If $\lim_{t \to \infty} \exp \left( \mathcal L \, t  \right) = \mathcal D_{B'}$ then $\exp\left( \mathcal L \, t \right)$ is unital for $t \ge 0$ and $\mathcal L (\ket{i'}\bra{j'}) = \lambda_{ij} \ket{i'} \bra{j'}$, with $\lambda_{ij} = - \im (E_{i} - E_{j}) + \sum_\alpha \left( \left(L_\alpha \right)_{ii}  \left(L_\alpha\right)_{jj}^* - 1/2 \left| \left(L_\alpha \right)_{ii}   \right|^2 -1/2  \left| \left(L_\alpha \right)_{jj}   \right|^2\right)\,$, where $E_i = \braket{i' | H | i'}$ and $\left(L_\alpha\right)_{ii} =  \braket{ i' |  L_\alpha |i'}$.
    \item Let in addition $B = \{ P_i \}_{i=1}^d$ and $U \in U(d)$ be a unitary operator such that $\ket{i'} = U \ket{i}$ for all $i=1,\dotsc,d$. If $\lim_{t \to \infty} \exp \left( \mathcal L \, t  \right) = \mathcal D_{B'}$, then
    \begin{gather}
       C_{2,B}\left[ \exp \left( \mathcal L \, t \right) \right] = \frac{1}{d(d+1)} \left[ \Tr\left( X_U \Lambda (t) X_U^T  \right) - \Tr \left( Y_U (t) \, Y_U^T (t) \right) \right] \,\;, \label{Dephasing_CGP_2norm_time}
    \end{gather}
    where $X_U \in  \mathbb R^{d \times d}$ is bistochastic  with $\left(X_U\right)_{ij} = |\!\braket{i |U| j}\!|^2$, $\Lambda (t),Y_U(t)  \in \mathbb R^{d \times d}$ with $\left[\Lambda (t)\right]_{i j}  = \exp \left( 2 \Real (\lambda_{ij})\, t  \right)$ and $\left[ Y_U(t) \right] _{ij} = \sum_{k,l} \exp \left( \lambda_{kl} t  \right) \, U_{il} U_{ik}^* U_{jk} U_{jl}^*\,$.
  \end{enumerate}
\end{prop}
\begin{proof}
  \textbf{(i) \& (ii)} The condition $\lim_{t \to \infty} \exp \left( \mathcal L \, t  \right) = \mathcal D_{B'}$ is equivalent to \begin{inparaenum}[$(a')$]
  \item $ P'_i \in \Ker \left(\mathcal L\right)$ for all $i$ (i.e., all $P'_i$ belong to the set of steady states) and
  \item all matrix elements $\braket{i'  | \exp\left(\mathcal L t \right) \, \rho_0 | j'}$ ($i \ne j$) vanish for $t \to \infty$ for all initial states $\rho_0$.
  \end{inparaenum}
  We will first show that $(a')$ is equivalent to $(a)$ \& $(b)$ holding true. Indeed, from the Lemma it follows that $\mathcal L P'_i = 0$ for all $i$ iff $H$ and $\{L_\alpha\}_\alpha$ are all diagonal with respect to $B$, i.e., conditions $(a)$ and $(b)$ are true. Now we need to make sure that $(b')$ holds, i.e.,  non-diagonal elements decay for $t \to \infty$. By plugging in the form of the Lindbladian Eq.~\eqref{Lindblad} it follows that $\mathcal L (\ket{i'}\bra{j'}) = \lambda_{ij} \ket{i'} \bra{j'}$, with
  \begin{align*}
  \lambda_{ij} = - \im (E_{i} - E_{j}) &+ \sum_\alpha \Big( \left(L_\alpha \right)_{ii}  \left(L_\alpha\right)_{jj}^* \\
  &- 1/2 \left| \left(L_\alpha \right)_{ii}   \right|^2  -1/2 \left| \left(L_\alpha \right)_{jj}   \right|^2\Big) \,\;,
  \end{align*}
  since all operators are diagonal in the $B'$ basis. Thus these elements decay iff $\Real (\lambda_{ij}) < 0$  for $i \ne j$, which is equivalent to  condition $(c)$, since $\Real \left( \lambda_{ij} \right) = -1/2 \sum_\alpha \left| \left( L_\alpha \right)_{ii} - \left( L_\alpha \right)_{jj}\right|^2$. Finally, the channel is unital for all $t \ge 0$ since $ I/d = \sum_i  P'_i/d \in \Ker \mathcal L$ as convex combination of elements of $B'$. \\ \\
  \textbf{(iii)}  Our starting point is Eq.~\eqref{CGP_unitals} (since $\mathcal E(t) = \exp\left( \mathcal L \, t \right)$ is unital $\forall t \ge 0$). We first calculate $\sum_i \braket{\mathcal E(t) P_i , \mathcal E(t) P_i }$. From (ii), $\mathcal E(t) (\cdot) = \sum_{k,l} P'_k (\cdot) P'_l \exp(\lambda_{kl}\, t)$. Using that, it is direct to show that
  \begin{align*}
  \sum_i \braket{\mathcal E(t) P_i , \mathcal E(t) P_i } &= \sum_{i,k,l} \left( X_U \right)_{ik} \exp\left[ (\lambda_{kl} + \lambda_{lk}) t  \right] \left( X_U \right)_{il} \\ &= \Tr\left( X_U \Lambda (t) X_U^T  \right) \,\;.
  \end{align*}
  In a similar way it follows that
  \begin{gather*}
  \sum_i \braket{\mathcal D_B \mathcal E(t) P_i ,\mathcal D_B \mathcal E(t) P_i } = \Tr \left( Y_U (t) \, Y_U^T (t) \right) \,\;.
  \end{gather*}
  Combining the two calculations the claimed result follows.
\end{proof}

Observe that for $t \to \infty$ expression $\eqref{Dephasing_CGP_2norm_time}$ reduces to Eq.~\eqref{dephasing_2norm_CGP}, as expected from the fact that $\lim_{t \to \infty} \exp\left( \mathcal L t \right) = \mathcal D_{B'}$. Indeed, under the conditions stated in Prop.~\autoref{maximal_dephasing_lindblad_prop}, $\Lambda (t \to \infty)  = I$ and $Y (t \to \infty) = X_U X_U^T$.

\subsection{Maximum Coherence Generating Power of dephasing Lindbladians} \label{Lindblad_maximal_section}

A natural question to be asked is whether or not there exist maximally dephasing Lindbladians which, although dephasing in the long time limit, have better capability to produce coherence for finite times than any maximally dephasing channel. As we will show momentarily, such dephasing time evolutions exist with 2-norm CGP that can get arbitrarily close to the maximum possible (over all unital quantum operations) $C_{2,B}$.

Consider for simplicity Lindbladian dynamics with Hamiltonian $H = 0$ and a single unitary Lindblad operator $V$, expressed as $V = \e ^{-\im H_V}$.  The evolution equation then takes the simple form $\mathcal L_V \,\rho = V \rho V^\dagger - \rho$. In the case where $H_V$ is non-degenerate, all conditions of prop.~\autoref{maximal_dephasing_lindblad_prop}(i) are met so in the long time limit  maximal dephasing occurs (with respect to the eigenbasis of the operator $H_V$).

Without loss of generality, we assume that all eigenvalues of $H_V$ are non-negative with $\left\| H_V \right\|_\infty < 2\pi$. Now let us examine what happens when $\left\| H_V \right\|_\infty \ll 1$. By expanding, we get $V = I - \im H_V + \mathcal O(\left\|  H_V \right\|_\infty^2)$, therefore
\begin{gather}
  \mathcal L_V \,\rho = - \im \left[ H_V , \rho \right] + \mathcal O\left( \left\|  H_V \right\|_\infty^2 \right) \,\;,
\end{gather}
As a result, for time scales such that $t \,\left\|  H_V \right\|_\infty^2 \ll 1$ the system undergoes ``almost'' unitary evolution under effective Hamiltonian $H_V$ and error $\mathcal O (t \left\|  H_V \right\|_\infty^2)$, which can be made arbitrarily small assuming $\left\|  H_V \right\|_\infty \ll 1$. The occurring effective unitary evolution is the key aspect that allows achieving a large CGP value, while dephasing is still the dominant process for large timescales $t \,\left\|  H_V \right\|_\infty^2 \gg 1$. Let us now make the above remarks precise. For the following, we normalize
\begin{gather}
  \tilde C_{2,B}  \left( \mathcal E \right) \coloneqq \frac{C_{2,B} \left( \mathcal E \right)}{ 4 C_{2,B}^{max}(d)}
\end{gather}
so that $0 \le \tilde C_{2,B} \left( \mathcal E \right) \le 1$ for all unital channels $\mathcal E$.

\begin{figure}[t]
  \centering
  \includegraphics[width= \columnwidth]{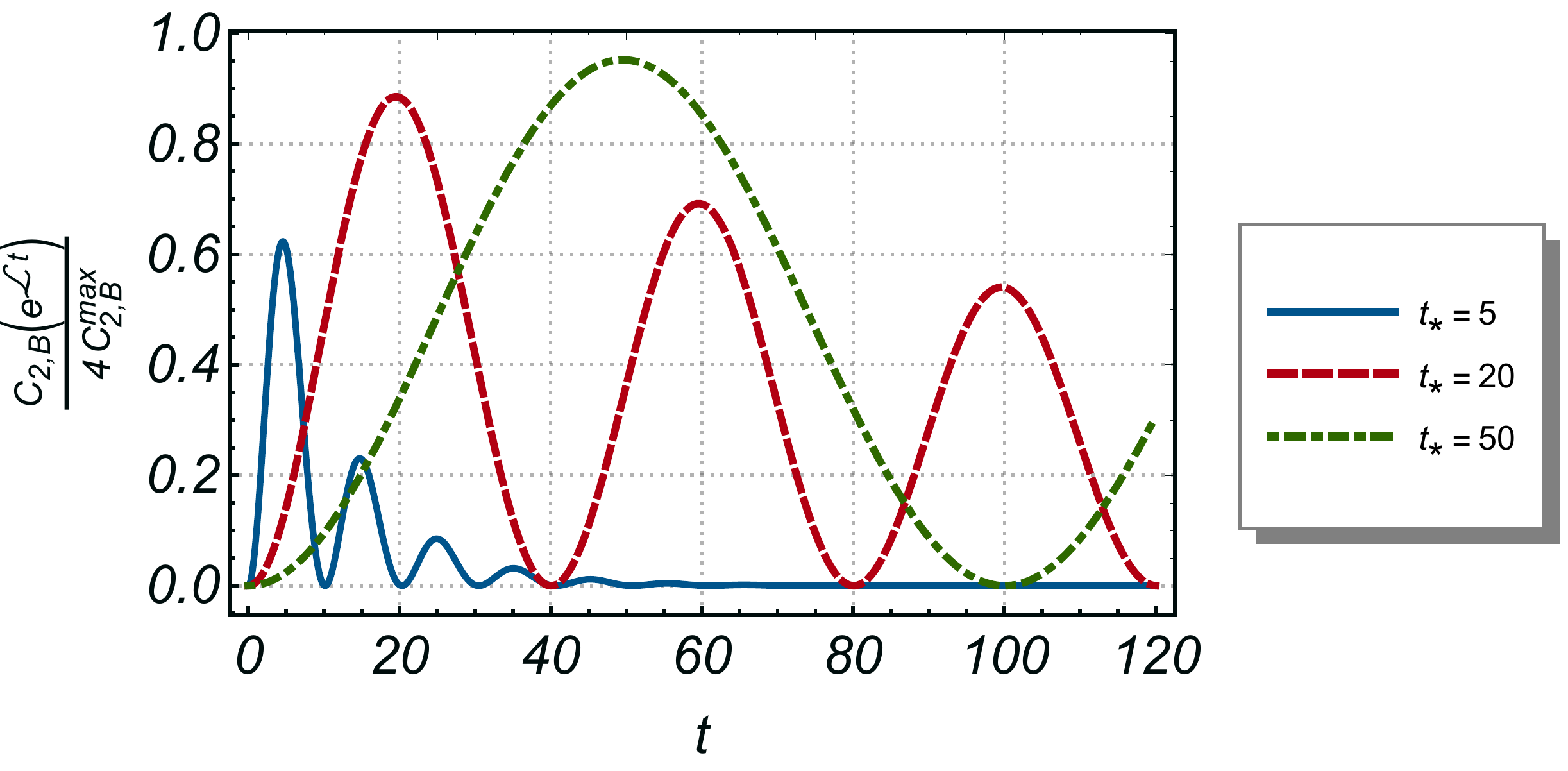}
  \caption{Plot of $C_{2,B}\left( \e ^{\mathcal L_V t} \right)$  for a 2-level system (normalized such that $4 C_{2,B}^{max} = 1$). The Lindbladian is chosen as in part (iii) of Prop.~\autoref{prop_phases}: $B = \{ P_0, P_1 \}$ and $V = P_+ + \e^{- \im \frac{\pi}{2 t_*} } P_-$ (where $P_\pm = \ket{\pm} \bra{\pm}$). Observe how, as $t_*$ increases (or, equivalently, $\left\| H_V \right\|_\infty$ decreases) the peak moves higher up, approaching unity.}\label{lindblad_graph}
\end{figure}

\begin{prop} \label{prop_phases}
  Let $B = \{ P_i \coloneqq \ket{i}\bra{i}\}_{i=1}^d$  be a complete family of \mbox{rank-1} orthogonal projectors and $V  = \e ^{- \im H_V}$ be a non-degenerate unitary with $\mathcal L_V (\cdot) = V (\cdot) V^\dagger - (\cdot)$ denoting the associated (maximally) dephasing Lindbladian and $\mathcal H_V \left( \cdot \right) \coloneqq - \im  \left[ H_V , \left( \cdot \right) \right]$ denoting the relevant Hamiltonian generator. Then,
  \begin{enumerate}[(i)]
    \item The difference of the CGP between the Lindbladian and the Hamiltonian evolution is bounded by
          \begin{align}\label{bound_unitary_lindblad_CGP}
            \left| \tilde C_{2,B} \left( \e ^{\mathcal L_V t} \right) - \tilde C_{2,B} \left( \e ^{\mathcal H_V t} \right) \right| \le 64 \frac{d}{d-1}\left\|  H_V \right\|_\infty^2 t
          \end{align}
    for any $\left\|  H_V \right\|_\infty \le 1/2$ and $t\ge 0$.
    \item Let $W$ denote a non-degenerate unitary connecting $B$ with a mutually unbiased basis, i.e.,
    \begin{gather*}
     \left| \braket{i | W | j} \right| = \frac{1}{\sqrt{d}} \;\, \forall i,j \,\;.
    \end{gather*}
    Then for $V =  W ^{1/t_*}$ and any $t_*\ge 4 \pi$,
        \begin{align}
         \left|  \tilde C_{2,B} \left( \e^{\mathcal L_V t_*} \right) -1 \right| \le 256 \pi^2 \frac{d}{d-1} \frac{1}{t_*} \,\;.
        \end{align}
    \item Let $F$ denote the quantum Fourier transform matrix, i.e.,
    \begin{align}
      \braket{j | F | k} = \frac{1}{\sqrt{d}} \exp \left(\im \frac{2\pi}{d} (j-1) \, (k-1)\right) \,\;.
    \end{align}
    If $H_V(\theta_d) = \sum_{k=1}^d \theta_k P'_{k} \,$, where $P'_{k} = F P_k F^\dagger$, $\theta_k = \theta_d \displaystyle \frac{f_k}{f_d}$ with $f_{k} = (k-1)(k-2)$ ($d$ odd) and $f_{k} = (k-1)^2$ ($d$ even), then
    \begin{align}
      \left| \tilde C_{2,B} \left( \e^{\mathcal L t_*} \right) -  1   \right| \le  64 \pi (d-1) \theta_d \,\;,  \label{maximum_CGP_lindblad_ineq}
    \end{align}
    for $ t_* (\theta_d) =\displaystyle \frac{\pi f_d}{d \theta_d}$ and $\theta_d \le \displaystyle   \frac{1}{2}$.
  \end{enumerate}
\end{prop}
\begin{proof}
   \textbf{(i)} We split the proof into three parts (a) -- (c) which can be combined to show the desired inequality. \\ \\
   \textbf{(a)} For unital CPTP maps $\mathcal E_1, \mathcal E_2$ the following inequality holds:
   \begin{align}\label{claim_1}
     \left| C_{2,B} \left(\mathcal E_1\right) - C_{2,B} \left(\mathcal E_2\right) \right| \le \frac{8}{d+1} \left\| \mathcal E_1 - \mathcal E_2 \right\|_{\diamondsuit} \,\;.
   \end{align}
   To show this, we start from Eq.~\eqref{CGP_unitals}. Using the triangle inequality, we get
   \begin{align*}
     & \big| C_{2,B} \left(\mathcal E_1\right)  - C_{2,B} \left(\mathcal E_2\right) \big|  \le \frac{1}{d (d+1)} \left( T_1 - T_2 \right) \,\;, \text{ where} \\
     & T_1  \coloneqq    \big| \textstyle \sum_{i} \displaystyle \braket{\mathcal E_1 P_i,\mathcal E_1 P_i } - \textstyle \sum_{i} \displaystyle \braket{\mathcal E_2 P_i,\mathcal E_2 P_i }  \big| \\
     & T_2  \coloneqq \big| \textstyle \sum_{i} \displaystyle \braket{\mathcal D_B \mathcal E_1 P_i,\mathcal D_B \mathcal E_1 P_i } -  \textstyle \sum_{i} \displaystyle \braket{\mathcal D_B \mathcal E_2 P_i,\mathcal D_B \mathcal E_2 P_i }  \big|  \,\;.
   \end{align*}
   Denoting $\rho_B \coloneqq 1/d \sum_{i=1}^d P_i \otimes P_i$ and using the identity
  \begin{align}\label{swap_identity}
    \Tr\left( A B \right) &= \Tr \left( P_{(12)} \, A \otimes B \right) \,\;,
  \end{align}
   where $P_{(12)} \coloneqq \sum_{i,j} \ket{ij} \bra{ji}$ is the SWAP operator, we get
   \begin{align*}
     T_1 & = d \left| \Tr \left( P_{(12)} \mathcal E_1 ^{\otimes 2} \rho_B  \right)-  \Tr \left( P_{(12)} \mathcal E_2 ^{\otimes 2} \rho_B  \right) \right|\\
         & = d \left| \Tr \left[ P_{(12)} \left(  \mathcal E_1 ^{\otimes 2} -  \mathcal E_2 ^{\otimes 2}\right) \rho_B \right] \right| \\
         & \le d  \left\| \left(  \mathcal E_1 ^{\otimes 2} - \mathcal E_2 ^{\otimes 2} \right) \rho_B  \right\|_1   \,\;,
   \end{align*}
   where in the third line we used the fact that $ \left| \Tr \left( A B \right) \right| \le \left\| A \right\|_\infty \left\| B \right\|_1$ and that $\left\| P_{(12)} \right\|_\infty = 1$. Now, since $\left\| \rho_B \right\|_1 = 1$, we have
   \begin{align*}
     T_1 \le d \left\|  \mathcal E_1^{\otimes 2}- \mathcal E_2^{\otimes 2}\right\|_{1,1} \le d \left\|  \mathcal E_1^{\otimes 2}- \mathcal E_2^{\otimes 2}\right\|_{\diamondsuit}  \,\;.
   \end{align*}
   Setting $\mathcal M \coloneqq \mathcal E_1 - \mathcal E_2$, we have
   \begin{align*}
     T_1 &\le d \left\| \left(\mathcal M + \mathcal E_2  \right)^{\otimes 2} - \mathcal E_2 ^{\otimes 2}   \right\|_{\diamondsuit}\\
         &\le d \left(  \left\| \mathcal M \otimes \mathcal M  \right\|_{\diamondsuit} + \left\| \mathcal M \otimes \mathcal E_2  \right\|_{\diamondsuit} + \left\| \mathcal E_2 \otimes \mathcal M  \right\|_{\diamondsuit}\right) \\
         &\le d \left( \left\| \mathcal M \right\|_{\diamondsuit}^2 + 2 \left\| \mathcal M \right\|_{\diamondsuit} \right) \le d \left\| \mathcal M  \right\|_{\diamondsuit} \left( \left\| \mathcal M  \right\|_{\diamondsuit} + 2  \right)\\
         & \le 4d \left\|  \mathcal M \right\|_{\diamondsuit} \,\;.
   \end{align*}
   Now we will bound the term $T_2$. The proof proceeds in a similar way. Since $\left\| \mathcal D_B \right\|_{1,1} = 1$, we also get
   \begin{align*}
     T_2 \le 4 d \left\|  \mathcal M \right\|_{\diamondsuit} \,\;.
   \end{align*}
   The desired inequality follows. \\ \\
   \textbf{(b)} We are going to prove that the following inequality holds:
   \begin{align}\label{claim_2}
     \left\| \e^{\mathcal L t} - \e^{\mathcal H t}   \right\|_{\diamondsuit} \le t \left\| \mathcal L - \mathcal H \right\|_{\diamondsuit}
   \end{align}
   for all $t\ge 0$. Denoting $\mathcal U_t \coloneqq \exp \left( \mathcal H t \right)$ we have
   \begin{align*}
     \left\| \e^{\mathcal L t} - \e^{\mathcal H t}   \right\|_{\diamondsuit} = \left\| \mathcal U^\dagger_t \e^{\mathcal L t}   - I   \right\|_{\diamondsuit} = \left\| \mathcal K_t   - I   \right\|_{\diamondsuit}
   \end{align*}
   where $\mathcal K_t \coloneqq \mathcal U^\dagger_t \exp\left(\mathcal L t\right) $. We also have $\dot {\mathcal K}(t) = \mathcal V_t \mathcal K(t)$ with $\mathcal V_t = \mathcal U^\dagger_t \left( \mathcal L - \mathcal H \right) \mathcal U_t$ (interaction picture). Now we can bound the quantity of interest:
   \begin{align*}
     \left\| \mathcal K_t   - I   \right\|_{\diamondsuit} &= \left\| \int_{0}^{t} ds \, \dot {\mathcal K_{s}}   \right\|_{\diamondsuit} = \left\| \int_{0}^{t} ds \, \mathcal V_{s} \mathcal K_{s}   \right\|_{\diamondsuit} \\
      & \le t \sup_{s \in [0,t]}  \left\|  \mathcal V_{s} \mathcal K_{s}   \right\|_{\diamondsuit} \le t \left\| \mathcal L - \mathcal H \right\|_{\diamondsuit} \,\;,
   \end{align*}
   where in the last inequality we used sub-multiplicativity and the unitary invariance of the diamond norm together with the fact that $\left\| \exp{ \left(\mathcal L t \right)} \right\|_{\diamondsuit} = 1$ (CPTP map for all $t \ge 0$).
   \\ \\
   \textbf{(c)}  For $\left\|  H_V \right\|_\infty \le 1/2$, we will show that
   \begin{align}\label{claim_3}
     \left\| \mathcal L_V - \mathcal H_V \right\|_{\diamondsuit} \le 8 \left\|  H_V \right\|_\infty^2 \,\;.
   \end{align}
   Notice that $H_V$ can always be chosen so that all eigenvalues are non-negative and $\left\|  H_V \right\|_\infty < 2\pi$. Now we define the superoperators
   \begin{align}
     \mathscr L_A \left(\rho\right) &= A \rho \\
     \mathcal R_B \left(\rho\right) &=  \rho B \,\;,
   \end{align}
   for which it holds that
   \begin{align*}
     \left\|  \mathscr L_A   \right\|_{\diamondsuit} &=  \left\|  \mathscr L_A  \otimes I_d \right\|_{1,1} =  \left\|  \mathscr L_{A\otimes I_d} \right\|_{1,1} \le \left\|  A \right\|_{\infty} \,\;, \\
     \left\|  \mathcal R_B   \right\|_{\diamondsuit} &=  \left\|  \mathcal R_B  \otimes I_d \right\|_{1,1} =  \left\|  \mathcal R_{B\otimes I_d} \right\|_{1,1} \le \left\|  B \right\|_{\infty}
   \end{align*}
   and therefore
   \begin{align} \label{diamond_ineq}
     \left\|  \mathscr L_A \mathcal R_B \right\|_{\diamondsuit} \le \left\| A \right\|_{\infty} \left\| B \right\|_{\infty} \,\;.
   \end{align}
   Setting $\Delta \coloneqq V -(I - \im H_V)$ we can express the action of $\left(\mathcal L_V - \mathcal H_V \right)$ on some $\rho$ as
   \begin{align*}
     \left( \mathcal L_V - \mathcal H_V \right) \rho = \Delta \rho + \rho \Delta ^\dagger &+ \Delta \rho \Delta^\dagger + H_V \rho H_V \\ &- \im H_V \rho \Delta^\dagger + \im \Delta \rho H_V \,\;,
   \end{align*}
   and therefore
   \begin{align*}
   \left\| \mathcal L_V - \mathcal H_V  \right\|_{\diamondsuit} \le  \left\| \mathscr L_{\Delta} \right\|_{\diamondsuit} + \left\| \mathcal R_{\Delta ^\dagger } \right\|_{\diamondsuit} &+  \left\| \mathscr L_{\Delta} \mathcal R_{ \Delta^\dagger} \right\|_{\diamondsuit}   + \left\| \mathscr L_{ H_V } \mathcal R_{ H_V} \right\|_{\diamondsuit} \\ & + \left\| \mathscr L_{H_V}  \mathcal R_{ \Delta^\dagger} \right\|_{\diamondsuit} + \left\|  \mathscr L_{\Delta} \mathcal R_{ H_V }\right\|_{\diamondsuit} \,\;.
   \end{align*}
   Using Eq.~\eqref{diamond_ineq}, the above reduces to
   \begin{align*}
     \left\| \mathcal L_V - \mathcal H_V \right\|_{\diamondsuit} \le 2 \left\| \Delta \right\|_\infty + \left\| \Delta \right\|^2_\infty + \left\| H_V \right\|_\infty^2  + 2 \left\| \Delta \right\|_\infty \left\| H_V \right\|_\infty
   \end{align*}
   We will estimate $\left\| \Delta \right\|_\infty$. We have
   \begin{align*}
     \left\| \Delta \right\|_\infty \le \sum _{n=2}^\infty \frac{\left\| H_V \right\|_\infty^n}{n!} \le \left\| H_V \right\|_\infty^2 \sum _{n=0}^\infty \frac{\left\| H_V \right\|_\infty^n}{(n+2)!} \,\;.
   \end{align*}
   Now we make the assumption that $\left\| H_V \right\|_\infty \le 1/2$. Under this assumption,
   \begin{align*}
     \left\| \Delta \right\|_\infty  \le \left\| H_V \right\|_\infty^2 \frac{1}{1 - \left\| H_V \right\|_\infty} \le 2 \left\| H_V \right\|_\infty^2 \,\;.
   \end{align*}
   Using this upper bound we get equation \eqref{claim_3}.

   Finally, we can combine together parts (a) -- (c) (and then normalize) to get the desired inequality for part (i) of the proposition. \\ \\
   \textbf{(ii)} We will first show that the unitary evolution has optimal CGP at $t=t_*$, namely $\tilde C_{2,B} \left( \e^{\mathcal H_V t_*} \right) = 1$.
   In \cite{coherence_1} it was shown that the maximum value of $\tilde C_{2,B}$ is $1$ and is attained by unitary channels $\mathcal U \left(  \cdot \right) = U \left( \cdot \right) U^\dagger$ iff
   \begin{gather}
   \left(X_U\right)_{ij} \coloneqq \left| \braket{i | U |j}\right|^2 = \frac{1}{d} \; \forall i,j \,\;. \label{condition_max_CGP}
   \end{gather}
   Here $U(t) = \exp \left( - \im H_V t \right)$. For $t = t_*$, we have $U(t_*) = V^{t_*} = W$ and thus, indeed, the above condition is satisfied.
   
   Now we can apply part $(i)$ of the proposition to get the desired bound. The matrix $W$ is unitary so $\left\| W \right\|_{\infty} \le 2 \pi$, hence an $H_V$ that satisfies the equation $\left(H_V\right)^{t_*} = W$ can be chosen with $\left\| H_V \right\|_\infty t_* \le 2\pi $. As a result, $\left\|  H_V \right\|_\infty \le 1/2$ from part $(i)$ implies $t_*\ge 4 \pi$.
   \\ \\
   \textbf{(iii)} We will first show that $\tilde C_{2,B} \left( \e^{\mathcal H_V t_*} \right) = 1$. As in the previous part, we need to prove that for the unitary operator $U(t_*) = \exp \left( - \im H_V t_* \right)$ Eq.~\eqref{condition_max_CGP} is satisfied. We have
   \begin{align}
   \left(X_U\right)_{ij} & = \left| \sum_{k = 1}^d  \e^{-\im \theta_k t_*}
    F_{ik} F_{jk}^* \right|^2 \\
    & = \frac{1}{d^2} \left|  \sum_{k=1}^{d} \exp\left(\im \frac{2\pi}{d} (k-1) (j-i)\right)  \exp\left(\im \frac{\pi}{d}f_{k}\right) \right|^2 \,\;.
   \end{align}
   Now consider the odd $d$ case. Substituting for $f_k$, we get
   \begin{align*}
     \left(X_U\right)_{ij} & = \left|  \sum_{k=0}^{d-1} \exp\left[\im \frac{\pi}{d} \left( k^2 +k [2 (j-i) - 3] \right) \right] \right|^2  \,\;,
   \end{align*}
   where the sum inside the modulus is a (generalized) quadratic Gauss sum (see, e.g., \cite{berndt1998gauss}) and for $d$ odd and $(j-i)$ integer, evaluates to $\sqrt{d}$ (ignoring some irrelevant phase factor). Therefore, indeed, $\left(X_U\right)_{ij} = 1/d$. The even $d$ case proceeds similarly.

   Now we can use part $(i)$ of the proposition for $t=t_*$, where $\left\| H_V \right\|_\infty = \theta_d$. Substituting for $t_*(\theta_d)$ and using the fact that $f_d \le (d-1)^2$ (for all $d$) we get the desired inequality.
\end{proof}
Prop.~\autoref{prop_phases} provides two different ``recipes'' to construct dephasing Lindbladians such that for $t=t_*$ the 2-norm CGP is nearly maximal and further provides an upper bound for the difference between the occurring CGP and the optimal one at $t=t_*$. An example demonstrating the construction described in part (iii) of the proposition is plotted in \autoref{lindblad_graph}. Notice that this family of Lindbladians can get arbitrarily close to the maximum value of CGP for $t = t_*$ but, nevertheless, has vanishing CGP for $t \to \infty$. This is because $\lim_{t \to \infty} \exp \left(\mathcal L_V \, t\right) = \mathcal D_{B'}$ with $B$ and $B'$ being mutually unbiased bases.


\section{Random dephasing channels} \label{random_section}

Now we investigate the situation where the maximally dephasing channels are \textit{random}. More specifically, in view of Eq.~\eqref{dephasing_2norm_CGP}, the CGP of a maximally dephasing channel can also be treated as a random variable over the unitary group $U(d)$, which we consider equipped with the Haar measure. In other words, the basis over which the quantum system is being dephased can be regarded as random variable as, for example, in the occurrence of an orthogonal measurement of some (\textit{a priori} unknown) non-degenerate observable. For the following, we use the normalization $\tilde C_{2,B}  \left( \mathcal D_{B'} \right) \coloneqq C_{2,B} \left( \mathcal D_{B'} \right) / C_{2,B}^{max}(d)$ so that $0 \le \tilde C_{2,B} \left( \mathcal D_{B'} \right) \le 1$ \footnote{For the results that follow to be meaningful for large dimension $d$, it is necessary that the normalization function $C_{2,B}^{max}(d)$ is not only an upper bound of $C_{2,B}(\mathcal D_{B'})$ (over all maximally dephasing channels), but also a maximum (i.e., is achievable). In the Appendix, we are able to prove that the upper bound is achievable for all Hilbert space dimensions up to $d = 13$.}.

\begin{prop}[CGP of random dephasing] \label{PDF_proposition}
   Let $P_{CGP} (c) \coloneqq \int d\mu_{Haar}(U) \, \delta\left( c - \tilde C_{2,B}\left( \mathcal D_{B'} \right) \right) $ be the probability density function of the  maximal dephasing (2-norm) coherence generating power over Haar distributed $U \in U(d)$. Then
\begin{enumerate}[(i)]
  \item For a qubit ($d=2$) the probability density function is
  \begin{gather}
    P_{CGP} \left( c \right) = \frac{1}{\sqrt{32 \, c (1-c)}} \left(  \sqrt{1+\sqrt{1-c}} + \sqrt{1-\sqrt{1-c}} \right) \,\;.
  \end{gather}

  \item The mean value
  $\braket{\tilde  C_{2,B} \left( \mathcal D_{B'} \right) }_U \coloneqq \int dc \,c P_{CGP}(c) $ is bounded from above by
  \begin{gather}
     \braket{ \tilde C_{2,B} \left( \mathcal D_{B'} \right) }_U  \le M(d) \,\;,  \label{mean_bound}
  \end{gather}
  where
  \begin{gather}\label{M_expression}
    M(d) \coloneqq \frac{4 d \left[d (d+5)+2\right]}{(d+1)^2 (d+2)(d+3)} \,\;.
  \end{gather}

  \item Using Levy's lemma for Haar distributed unitary matrices, we obtain
  \begin{gather}\label{Levy_dephasing}
          \Prob \Bigg\{ \tilde C_{2,B} \left( \mathcal D_{B'} \right) \ge \frac{1}{d^{1/4}} + M(d) \Bigg\} \le \exp \left( - \frac{\sqrt{d}}{640^2} \right) \,\;.
  \end{gather}
\end{enumerate}
\end{prop}
\begin{proof}
  \textbf{(i)} In Ex.~\autoref{qubit_example_1} the 2-norm CGP for a qubit (in the Bloch sphere parametrization) was found to be $\tilde C_{2,B} (\theta) = \sin^2\left( 2 \theta \right)$. In this parametrization, Haar distributed unitary matrices $U \in U(2)$ correspond to the  measure $ 1 / (4\pi) \int_0^{2\pi} d\phi \int_0 ^\pi d\theta \sin \theta $. Therefore $P_{CGP}(c) = 1 / (4\pi) \int_0^{2\pi} d\phi \int_0 ^\pi d\theta \sin \theta \, \delta \left( c -  \sin^2\left( 2 \theta \right) \right)$. The result follows directly by performing the integral (e.g., by changing variables). \\ \\
  \textbf{(ii)} We have to average Eq.~\eqref{dephasing_2norm_CGP} over $U$. By linearity, we can calculate $\Braket{\Tr \left( X_U X_U^T \right)}_U$ and $\Braket{\Tr \left[ (X_U X_U^T)^2 \right]}_U$ separately.

  The first quantity being averaged is equal to $\Tr \left( X_U X_U^T \right) = \sum_{i,j} \left| \braket{i | U |j} \right| ^4$. By setting $\ket{j_U} \coloneqq U \ket{j}$, we have $\left| \braket{i | j_U} \right|^4 = \Tr \left( \left(\ket{i} \bra{i}\right)^{\otimes 2}  \left(\ket{j_U }\bra{j_U }\right)^{\otimes 2}  \right)$. Again by linearity, it follows that $\Braket{\sum_{i,j} \Tr \left( \left(\ket{i} \bra{i}\right)^{\otimes 2}  \left(\ket{j_U }\bra{j_U }\right)^{\otimes 2}  \right)}_U = \sum_{i,j} \Tr \left(  \left(\ket{i} \bra{i}\right)^{\otimes 2}  \Braket{\left(\ket{j_U }\bra{j_U }\right)^{\otimes 2}}_U \right)$. Now we can employ the well-known general result (for a proof see, e.g., \cite{goodman2000representations})
  \begin{gather}
  \Braket{\left(\ket{j_U}\bra{j_U} \right)^{\otimes n}}_U = \frac{1}{n!} \frac{1}{{d+n-1\choose n}}  \sum_{\pi \in S_n} P_\pi\,\;,
  \end{gather}
  where $S_n$ is the symmetric group of $n$-objects and $P_\pi$ is the operator that enacts the permutation $\pi$ in $\mathcal H^{\otimes n}$. For $n=2$, we have
  \begin{gather}
  \Braket{\left(\ket{j_U}\bra{j_U} \right)^{\otimes 2}}_U = \left[d(d+1)\right]^{-1}\left(I + P_{(12)}\right)\,\;,
  \end{gather}
  where $P_{(12)}$ is the $(12)$ cycle (i.e., $P_{(12)}$ is just the SWAP operator). Plugging this in and performing the trace, we get $\Braket{\Tr \left( X_U X_U^T \right)}_U = 2d/(d+1)$.

\begin{figure}[t]
  \centering
  \includegraphics[width=\columnwidth]{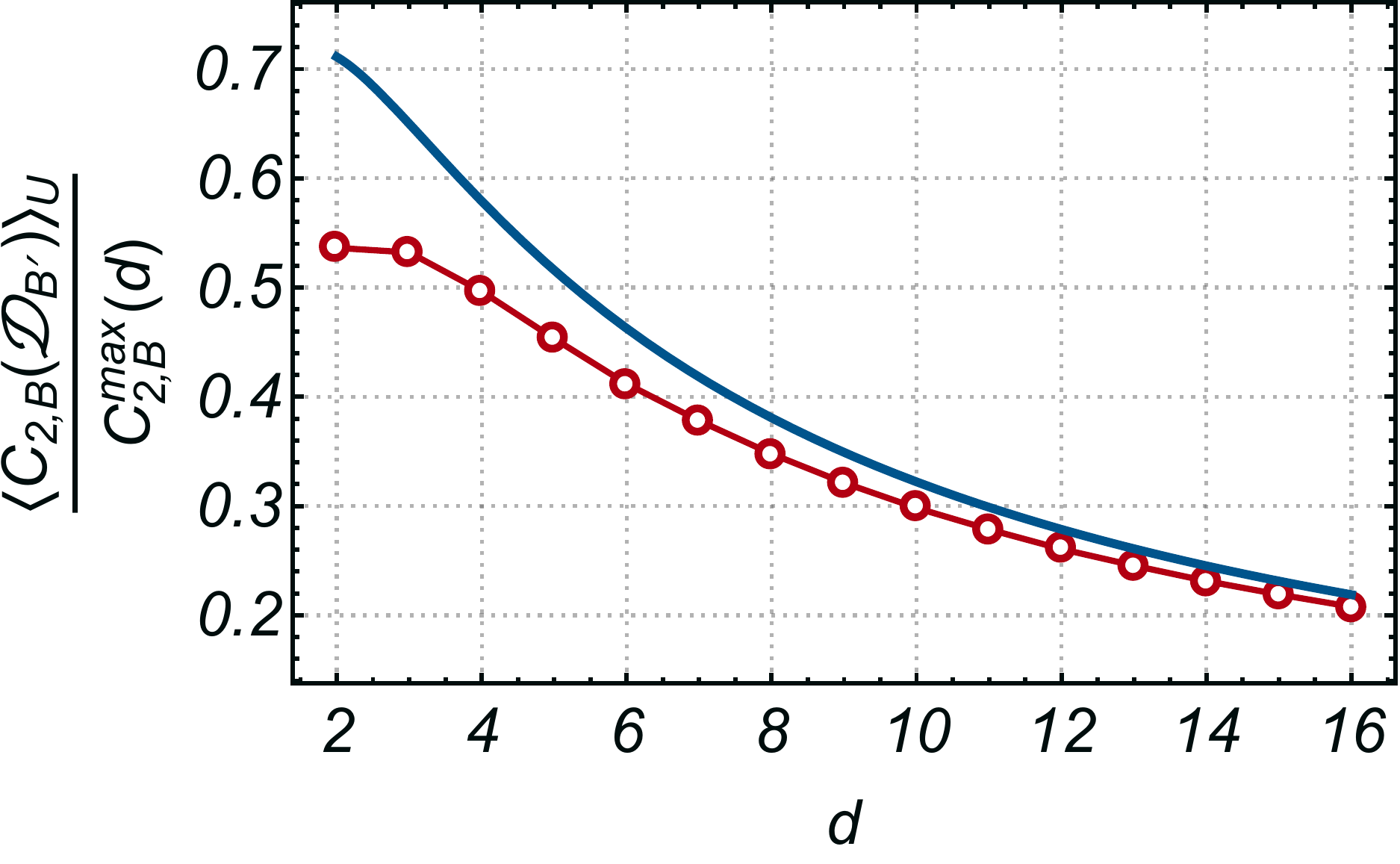}
  \caption{Comparison between the numerically computed dephasing CGP mean $\braket{ C_{2,B} \left( \mathcal D_{B'} \right) }_U$ (individual points) and its upper bound from Eq.~\eqref{mean_bound} (solid line), as a function of the Hilbert space dimension $d$. Both quantities are normalized by dividing with the upper bound $C_{2,B}^{max}(d)$.
  }\label{upper_bound}
\end{figure}

  We will follow a similar strategy for the second quantity $\Braket{\Tr \left[ (X_U X_U^T)^2 \right]}_U$. The quantity being averaged can be reexpressed as
  \begin{align*}
  \Tr \left[ (X_U X_U^T)^2 \right] &= \sum_{i,j,k,l} \left| \braket{i | k_U} \right|^2\left| \braket{j | k_U} \right|^2\left| \braket{i | l_U} \right|^2\left| \braket{j | l_U} \right|^2 \\
  &= \sum_{i,j,k,l} \Tr \left( \ket{iijj} \bra{iijj}  \left( \ket{k_U l_U} \bra{k_U l_U}\right)^{\otimes 2} \right) \,\;.
  \end{align*}
  Now we split the sum to two parts: $k=l$ and $k \ne l$, which we call $\Sigma_1$ and $\Sigma_2$, respectively. For $\Sigma_1$ we get $\sum_{i,j,k} \Tr \left( \ket{iijj} \bra{iijj}  \left( \ket{k_U} \bra{k_U}\right)^{\otimes 4} \right)$. Taking the average we can use the formula as before with $n=4$. Therefore we now have to evaluate $\sum_{i,j,\pi}  \Tr \left( \ket{iijj} \bra{iijj} P_\pi  \right)$ for all permutations $\pi \in S_4$. Out of the $4!=24$ elements, after performing the $i,j$ sum, $4$ of them give $d^2$ [the permutations with cycle decomposition $(12)$, $(34)$, $(12)(34)$ and the identity permutation] while the rest give $d$. As a result, $\Sigma_1 = (4d+20)/[(d+1)(d+2)(d+3)]$.

  So far everything is exact. For the $\Sigma_2$ term we notice we can write $\braket{\Sigma_2}_U = \sum_{k \ne l} \Braket{  \left(\Tr  \Bigl[  \left( \sum_i \ket{ii} \bra{ii} \right)\ket{k_U l_U} \bra{k_U l_U} \Bigr] \right)^2 }_U$. We know approximate the mean $\Braket{\Sigma_2}_U$ using the inequality $\braket{A^2} \ge \braket{A}^2$, which yields $\braket{\Sigma_2}_U \ge \sum_{k \ne l} \left( \Tr \left[  \left(\sum_i \ket{ii} \bra{ii} \right) \braket{\ket{k_U l_U} \bra{k_U l_U}}_U \right] \right)^2$. Now we cannot use the formula as before to calculate the average, since $\ket{k_U}$ and $\ket{l_U}$ are correlated. Nevertheless, we can use the slightly more general result (see, e.g., \cite{zanardi2014local})
  \begin{align}
  \braket{U^{\otimes 2} A U^{\dagger \otimes 2} }_U = &\left(\frac{ \Tr A}{d^2-1} - \frac{\Tr \left[P_{(12)}  A\right] }{d(d^2-1)} \right) I  \nonumber \\
  - &\left(\frac{ \Tr A}{d(d^2-1)} - \frac{\Tr \left[P_{(12)}  A\right] }{d^2-1} \right) P_{(12)}\,\;.
  \end{align}
  Evaluating for $A = \ket{k l} \bra{k l}$ (with $k \ne l$) we get $\Sigma_2 \ge d(d-1)/(d+1)^2$.

  Putting everything together we get
  \begin{align}
  M(d) = \frac{4d}{(d-1)(d+1)} \left( \frac{d+3}{d+1} - \frac{4d + 20}{(d+2)(d+3)} \right) \,\;,
  \end{align}
  which simplifies to the expression claimed. \\ \\
  \textbf{iii)} In order to prove the desired inequality we are going to use the following form of Levy's lemma for Haar distributed $U \in U(d)$ (see, e.g., \cite{anderson2010introduction}):
  \begin{gather}\label{Levy_general}
    \Prob \Big\{ \left| f(U) - \braket{f(U)}_U \right| \ge \epsilon \Big\} \le \exp \left( - \frac{d \epsilon^2}{4 K^2} \right) \,\;,
  \end{gather}
  where $K$ is a Lipschitz constant of $f: U(d) \to \mathbb R$, i.e., $\left| f(U) - f(V) \right| \le K \left\| U - V \right\|_2$. Here our function $f(U)$ is going to be $\tilde C_{2,B}\left( \mathcal D_{B'} \right)$ (viewed as a function of the unitary $U$ connecting $B$ and $B'$), i.e.,
  \begin{gather*}
    f(U) \coloneqq \frac{4}{d-1} \left( \Tr \left( X_U X_U^T \right) - \Tr \left[ \left( X_U X_U^T \right)^2 \right]  \right) \,\;.
  \end{gather*}
  Although the exact expression for the mean value $\braket{f(U)}_U$ has not been calculated, the upper bound from Eq.~\eqref{M_expression} allows approximating the desired probability, since
  \begin{align*}
    &\Prob \Big\{ \left| f(U) - \braket{f(U)}_U \right| \ge \epsilon \Big\}  \ge \\
    &\Prob \Big\{  f(U)  \ge \epsilon + \braket{f(U)}_U  \Big\}  \ge\\
    &\Prob \Big\{ f(U)  \ge \epsilon + M(d)  \Big\} \,\;.
  \end{align*}
  To complete the proof we need to estimate a Lipschitz constant for the function $f(U)$. We have
  \begin{align*}
   \left| f(U) - f(V) \right|   \le \frac{4}{d-1} \left( T_1 + T_2 \right) \,\;,
  \end{align*}
  where we set $T_1 \coloneqq  \left| \Tr\left( X_U X_U^T \right) - \Tr\left( X_V X_V^T \right) \right|$ and $T_2 \coloneqq  \left| \Tr\left[ \left(X_U X_U^T \right)^2 \right] - \Tr\left[ \left( X_V X_V^T\right)^2  \right] \right|$. From the proof of Prop.~\autoref{2norm_CGP_prop}, we can equivalently write
  \begin{align*}
      T_1 &= \left| \textstyle \sum_i \displaystyle \braket{\mathcal D_{B'(U)} P_i , \mathcal D_{B'(U)} P_i}  -  \textstyle \sum_j \displaystyle \braket{\mathcal D_{B'(V)} P_j , \mathcal D_{B'(V)} P_j}  \right| \\
          &= \left|  \textstyle \sum_i \displaystyle \braket{\mathcal D_{B} \mathcal U^\dagger P_i , \mathcal D_{B} \mathcal U^\dagger P_i}  -  \textstyle \sum_j \displaystyle \braket{\mathcal D_{B}\mathcal V^\dagger P_j , \mathcal D_{B} \mathcal V^\dagger P_j}  \right| \\
          & = \left|  \textstyle \sum_{i,j} \displaystyle \Tr \left( P_{(12)}  \mathcal D_B ^{\otimes 2} \left[ \mathcal U^{\dagger \otimes 2} P_i^{\otimes 2} - \mathcal V^{\dagger \otimes 2} P_j^{\otimes 2}  \right] \right)  \right| \,\;,
  \end{align*}
  where in the last step we used Eq.\eqref{swap_identity}. An upper bound for this quantity was calculated in \cite{coherence_1}, namely
  \begin{gather*}
    T_1 \le 8 d \left\| U -V  \right\|_2 \,\;.
  \end{gather*}
  Before proceeding with calculation of an upper bound for $T_2$, let us first prove the following inequality which will be needed momentarily:
  \begin{align}\label{norm_lemma}
    \left\|  U \rho U^\dagger - V \rho V^\dagger \right\|_1 \le 4 \left\| U - V  \right\|_\infty \,\;,
  \end{align}
  where $U,V$ are unitary and $\rho$ is an operator with $\left\|\rho  \right\|_1 = 1$. Setting $\Delta \coloneqq U - V$ we have $\left\|  U \rho U^\dagger - V \rho V^\dagger \right\|_1  = \left\| \Delta \rho \Delta^\dagger + \Delta \rho V^\dagger + V \rho \Delta ^\dagger \right\|_1 \le \left\| \Delta \rho \Delta^\dagger \right\|_1 + \left\| \Delta \rho V^\dagger \right\|_1  + \left\| V \rho \Delta ^\dagger \right\|_1$. Using the facts that the norm is unitarily invariant, $\| \Delta \|_\infty \le 2$ and that $\left\| A B  \right\|_1 \le \left\| A \right\|_1  \left\| B \right\|_\infty$ the aforementioned inequality follows.

  We can express $T_2$, in the same spirit as before, as
  \begin{align*}
    T_2 &= \Big| \textstyle \sum_i \displaystyle  \braket{ \mathcal D_{B}\mathcal D_{B'(U)} P_i , \mathcal D_{B} \mathcal D_{B'(U)} P_i} \\
     &\qquad \qquad \qquad -  \textstyle \sum_j \displaystyle \braket{\mathcal D_{B}\mathcal D_{B'(V)} P_j , \mathcal D_{B} \mathcal D_{B'(V)} P_j}  \Big| \\
    & = \Big|   \Tr \Big( P_{(12)} \mathcal D_B ^{\otimes 2} \Big[ \left( \mathcal U \mathcal D_B  \mathcal U^{\dagger} \right)  ^{\otimes 2} \left( \textstyle \sum_{i} \displaystyle P_i^{\otimes 2} \right) \\
     &\qquad \qquad \qquad  - \left( \mathcal V \mathcal D_B  \mathcal V^{\dagger} \right)  ^{\otimes 2}\left( \textstyle \sum_{j} \displaystyle P_j^{\otimes 2} \right)  \Big] \Big)  \Big| \\
    & = d \left| \Tr \left( P_{(12)} \mathcal D_B ^{\otimes 2} \left[ \left( \mathcal U \mathcal D_B  \mathcal U^{\dagger} \right)  ^{\otimes 2} \rho_B - \left( \mathcal V \mathcal D_B  \mathcal V^{\dagger} \right)  ^{\otimes 2} \rho_B  \right] \right)  \right| \,\;,
  \end{align*}
  where in the last step we set $\rho_B \coloneqq 1/d \sum_i P_i^{\otimes 2}$ (notice $\left\| \rho_B \right\|_1 = 1$). Now using the inequality $\Tr\left( A B \right) \le \left\| A \right\|_1 \left\| B \right\|_\infty $ (here $\left\| P_{(12)}  \right\|_\infty = 1$) and the fact that $\mathcal D_B$ is a CPTP map, we get
  \begin{align*}
    T_2 &\le d \left\|   \left( \mathcal U \mathcal D_B  \mathcal U^{\dagger} \right)  ^{\otimes 2} \rho_B  -  \left( \mathcal V \mathcal D_B  \mathcal V^{\dagger} \right)  ^{\otimes 2} \rho_B  \right\|_1 \\
    & = d \left\| \mathcal U^{\otimes 2} \rho_1 -  \mathcal V^{\otimes 2} \rho_2   \right\|_1 \\
    & = d \left\| \mathcal U^{\otimes 2} \rho_1 -  \mathcal V^{\otimes 2} \rho_1 +  \mathcal V^{\otimes 2} (\rho_1 - \rho_2)  \right\|_1  \\
    & \le d \left( \left\| \mathcal U^{\otimes 2} \rho_1 -  \mathcal V^{\otimes 2} \rho_1 \right\|_1 + \left\| \mathcal V^{\otimes 2} (\rho_1 - \rho_2)  \right\|_1   \right) \,\;,
  \end{align*}
  where we set $\rho_1 \coloneqq \left( \mathcal D_B \mathcal U^\dagger \right) ^{\otimes 2}  \rho_B$ and $\rho_2 \coloneqq \left( \mathcal D_B \mathcal V^\dagger \right) ^{\otimes 2}  \rho_B$. Now the inequality from Eq.~\eqref{norm_lemma} applies to both terms, yielding
  \begin{align*}
    T_2 &\le 8 d \left\| U^{\otimes 2} - V^{\otimes 2}  \right\|_\infty \\
        & = 8 d \left\| \left( V^\dagger \Delta + I \right)^{\otimes 2} - I \right\|_{\infty} \\
        &\le 8d \left( \left\| V^\dagger \Delta \right\|_\infty ^2 + 2 \left\| V^\dagger \Delta \right\|_\infty \right) \\
        & \le 8d \left\| \Delta \right\|_\infty \left( 2 + \left\| \Delta \right\|_\infty \right) \\
        & \le 32d \left\| U - V \right\|_\infty \le 32d \left\| U - V \right\|_2 \,\;.
  \end{align*}
  Finally, we obtain
  \begin{align*}
    \left| f(U) - f(V) \right| \le 160 \frac{d}{d-1} \left\| U - V \right\|_2 \,\;,
  \end{align*}
  therefore the Lipschitz constant can be taken to be  $K = 320$. The parameter $\epsilon$, in order to give a meaningful result for large Hilbert space dimension $d$, can be taken to be $\epsilon = d^{-\alpha}$, with $\alpha \in (0,1/2)$. Here we choose $\alpha = 1/4$. The inequality follows.
\end{proof}

\begin{figure}
  \centering
  \includegraphics[width=\columnwidth]{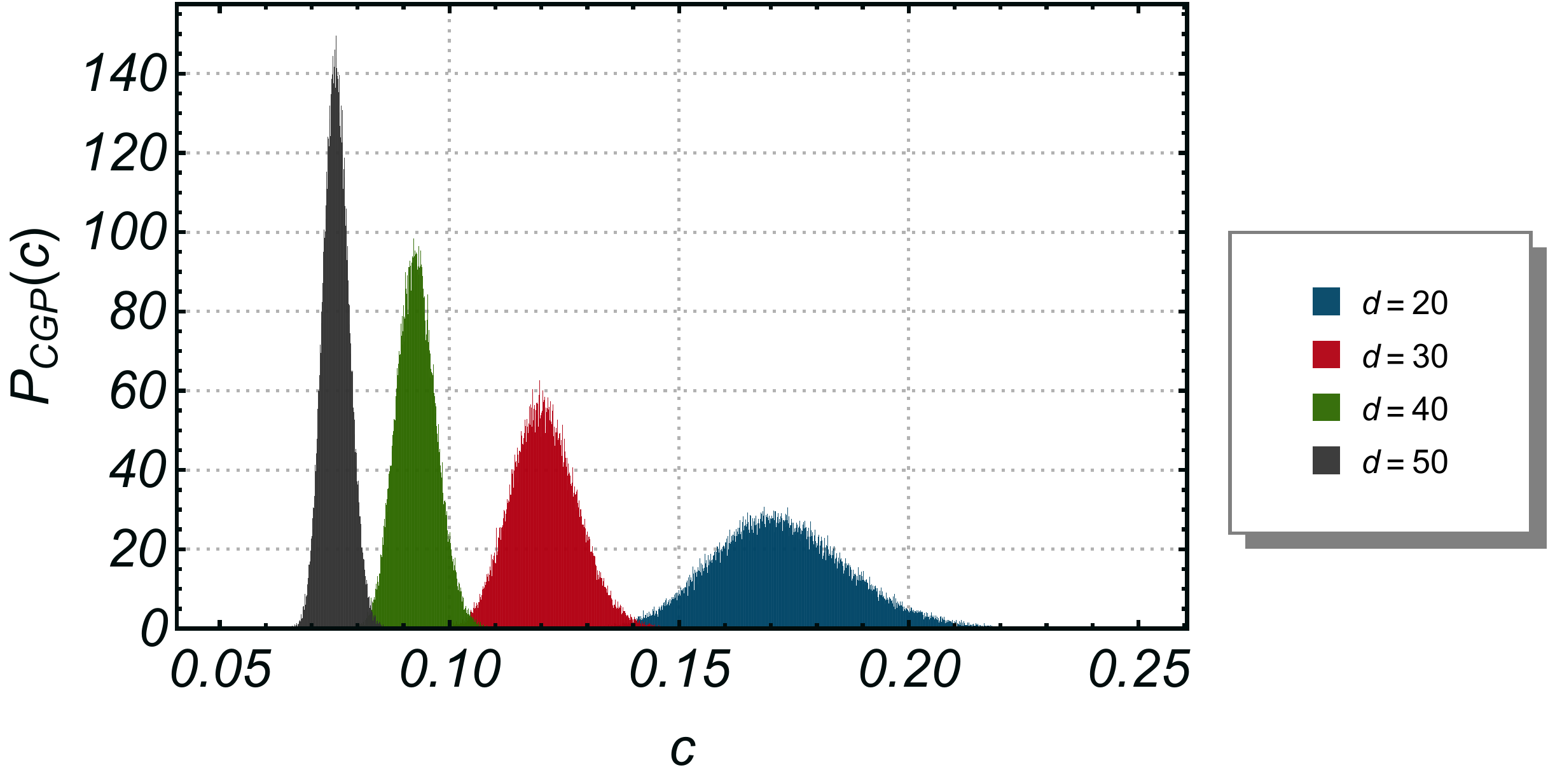}
  \caption{Numerically computed probability distribution functions $P_{CGP}(c)$ for the CGP of random maximally dephasing processes for different Hilbert space dimensions $d$. Last part of Prop.~\autoref{PDF_proposition} guarantees that, for sufficiently large $d$, the probability distribution function is concentrated around the mean value, which decreases as $d$ gets larger, as is indeed observed. Notice that in practice the concentration occurs for smaller $d$ that what is guaranteed by the proposition.}
  \label{histogram}
\end{figure}

Proposition \autoref{PDF_proposition} demonstrates that, in a quantum system described by a large Hilbert space, a maximally dephasing process over a random basis has vanishingly small capability to produce coherence out of incoherent states. Part (ii) of the above proposition sets an upper bound the (properly normalized) CGP, establishing that for large Hilbert space dimension $d$, the quantity $\braket{ \tilde C_{2,B} \left( \mathcal D_{B'} \right) }_U$ drops at least as fast as  $\sim 1/d$. A graphical comparison between the upper bound from Eq.~\eqref{mean_bound} and the numerically computed mean is presented in \autoref{upper_bound}. The last part of Prop.~\autoref{PDF_proposition} shows that, as the Hilbert space size grows, a maximally dephasing process occurring over a random basis has (with an exponentially increasing probability) CGP which is tightly distributed around the (decreasing) mean value. This concentration of the probability distribution function around the mean is depicted in \autoref{histogram}.

\section{Conclusion \& outlook} \label{conclusions_section}

In this work we have investigated the ability of various dephasing processes to generate coherence. For this purpose, we adopted various measures for the \textit{coherence generating power} of quantum channels, all based on probabilistic averages and arising from the viewpoint of coherence as a resource theory. We provided explicit formulas for maximally dephasing processes, valid for all finite Hilbert space dimensions, measuring how much coherence in generated on average from incoherent states when the Hilbert-Schmidt 2-norm and the relative entropy of coherence are used as quantifiers. In all cases, the coherence generating power of the dephasing process depends on the interplay between the bases over which coherence is quantified and dephasing occurs. This capability clearly vanishes when the two bases coincide while the maximum capability occurs for a basis which depends on the measure of state coherence. If the basis over which dephasing occurs is chosen at random in a uniform way, the average coherence generating power drops fast as the Hilbert space dimension increases.

We then extended the analysis to all Lindblad type quantum evolutions that maximally dephase in the infinite time limit by calculating the relevant Hilbert-Schmidt 2-norm coherence generating power of the associated time evolution for all intermediate times. Although maximally dephasing processes can generate finite amounts of coherence (depending on the associated bases), coherence generation cannot be as powerful as for some unitary processes. This is not always the case, however, for Lindblad evolutions that lead to dephasing. For the latter, we identified families of time propagators that have vanishing coherence generating power in the long time limit but nevertheless can get arbitrarily close to having optimal one for intermediate times.

The \textit{coherence generating power} of a quantum operation admits, directly by its definition, an interpretation as the average coherence contained in the post-processed states, as quantified by the relevant coherence measure. Nevertheless, an \textit{operational interpretation} of the \textit{coherence generating power}, relevant to practical tasks for which coherence is known to be a critical ingredient (such as those mentioned in the \autoref{intro_section}), is missing and could represent a challenge for future investigation.


\acknowledgements

G.S. thanks M.~Tomka and J.~Marshall for helpful discussions, and R.~Di Felice and the CNR-NANO Institute in Modena, Italy for their kind hospitality. This work was partially supported by the ARO MURI grant W911NF-11-1-0268.

\bibliography{coherence}

\appendix

\section{Attaintment of the upper bound $C_{2,B}^{max}(d)$} \label{appendix_maximum}

In this section we examine if the upper bound
\begin{align}
C_{2,B}^{max}(d) \coloneqq \frac{d-1}{4 d (d+1)}
\end{align}
of the maximally dephasing 2-norm CGP $C_{2,B}\left( \mathcal D_{B'} \right)$ (Eq.~\eqref{dephasing_2norm_CGP} of the main text) is attained over some basis $B'$. For example, for a qubit it can be explicitly verified (as in Ex.\autoref{qubit_example_1}) that the upper bound $C_{2,B}^{max}(d)$ is achievable. Here we tackle the general $d$-dimensional case.

From the proof of Prop.~\autoref{2norm_CGP_prop} it follows that the maximum value $C_{2,B}^{max}(d)$ for some (fixed) $d$ is attained if and only if there exists a unitary matrix $U$ such that $\sigma\left( X_U X_U^T \right) = \left\{ 1 , 1/2 \right\}$ with $1$ being a simple eigenvalue ($\sigma \left( A \right)$ denotes the spectrum of the operator $A$). Such a $d$-dimensional matrix is not guaranteed to exist a priori, since the $d$-dimensional unistochastic matrices are a proper subset of the $d$-dimensional bistochastic matrices for $d \ge 3$ (see, e.g., \cite{bengtsson2005birkhoff}).

For what follows, we further restrict to those bistochastic matrices $X_U$ such that \begin{inparaenum}[(a)]
\item are is symmetric and
\item have spectrum $\sigma(X_U) = \left\{ 1, 1/\sqrt{2} \right\}$.
\end{inparaenum}
Such a matrix has the form
\begin{gather*}
X_U = \ket{\psi_1}\bra{\psi_1} + \frac{1}{\sqrt{2}} \sum_{i=2}^d \ket{\psi_i}\bra{\psi_i} = \left( 1 - \frac{1}{\sqrt{2}} \right)\ket{\psi_1}\bra{\psi_1} + \frac{1}{\sqrt{2}} I \,,
\end{gather*}
where $\left\{ \ket{\psi_i} \right\}_{i=1}^d$ is the eigenbasis of $X_U$. However, $\left(X_U\right)_{ij}$ should also be a bistochastic matrix when expressed in the $B = \left\{ \ket{i} \right\}_{i=1}^d$ basis. This fixes the components to
\begin{gather}\label{unistochastic_components}
  \left( X_U\right)_{ij} = \frac{1}{\sqrt{2}} \delta_{ij} + \frac{1}{d} \left( 1 - \frac{1}{\sqrt{2}} \right) \,\;.
\end{gather}

The above $X_U$ matrix is circulant and therefore diagonalizable by the discrete Fourier transform \cite{smith2015unistochastic} $W_{lm} = \frac{1}{\sqrt{d}} \exp \left(\im \frac{2\pi}{d} (l-1) \, (m-1)\right) $. Now we further restrict to circulant $U$
which is hence also diagonalized by $W$. If such a unitary $U$ exists, is given by $U = W D W^\dagger$, where $D\coloneqq \diag \left( e^{\im \alpha_0}, \dotsc , \e^{\im \alpha_{d-1}} \right)$. As a result, Eq.~\eqref{unistochastic_components}, after some calculations, reduces to the following $(d-1)$ equations involving the eigenvalues of $U$:
\begin{gather}\label{phases_equation}
  \sum_{m=0}^{d-1} \exp\left[ \im \left( \alpha_{m+r} - \alpha_m\right) \right] = \frac{d}{\sqrt{2}} \,\;, \quad r=1,\dotsc, d-1 \,\;,
\end{gather}
where the index addition is understood $ Mod(d)$. The above set of equations for the eigenvalues of $U$ constitutes a sufficient condition for the attaintment of $C_{2,B}^{max}(d)$.

A family of solution to the above equations, valid for $d=2,\dotsc, 13$, is given by $\alpha_m = \phi_0$ for $m=0,\dotsc,k-1,k+1,\dotsc,d-1$ with $\phi_0 \in [0,2\pi)$ and $\alpha_k = \phi_0 + \phi$, where $\cos\phi =  \frac{1}{2} \left( \frac{d}{\sqrt{2}} -d +2 \right) $. The restriction $d \le 13 $ comes from $\cos \phi  \ge -1$.

\end{document}